\documentclass[12pt]{article}
\usepackage[utf8]{inputenc}
\usepackage{lmodern}
\usepackage{etoolbox}
\usepackage{enumerate}
\usepackage[T1]{fontenc}
\usepackage{mathrsfs}
\usepackage{graphicx}
\usepackage{amsmath,amsthm,amssymb}
\usepackage{amsfonts}
\usepackage{bbm}
\usepackage{setspace}
\usepackage{nicefrac}
\usepackage{bbm}
\usepackage{bibentry}
\usepackage[square,comma]{natbib}
\usepackage{caption}
\usepackage{tikz-cd}
\usetikzlibrary{positioning,calc}
\usetikzlibrary{arrows, arrows.meta, positioning, automata, decorations.pathreplacing, decorations.markings, patterns}
\usepackage{mathptmx}
\usepackage{eufrak}
\usepackage{float}
\usepackage{forest}
\usepackage{mathtools}
\usepackage{hyperref}
\hypersetup{
	colorlinks   = true, 
	urlcolor     = blue, 
	linkcolor    = blue, 
	citecolor   = blue 
}
\usepackage[defaultcolor=cyan]{changes} 
\definechangesauthor[color=red]{B}

\usepackage[a4paper,margin=3.2cm]{geometry}
\usepackage{titlesec}
\titleformat*{\section}{\large\bfseries}
\titleformat*{\subsection}{\bfseries}

\def\qed{\hfill $\Box$}                              	

\newtheorem{defn}{Definition}
\newtheorem{theorem}{Theorem}

\newtheorem{proposition}{Proposition}
\newtheorem{lemma}{Lemma}

\begin{document}
	
	\thispagestyle{empty}
	\def\thefootnote{\fnsymbol{footnote}}
	
	\begin{center}
		{\Large\bf Decisions over Sequences: Computability and Choice}\footnote{We are grateful to Debasis Mishra, Arunava Sen and Ariel Rubinstein for their guidance and encouragement. For useful comments, we thank Siddharth Barman, Umang Bhaskar, Pradeep Dubey, Bhaskar Dutta, Faruk Gul, Sean Horan, Ehud Lehrer, Stephen Morris, Yusufcan Masatlioglu, Herv\'{e} Moulin, Yuval Salant, Eran Shmaya, R Ramanujam, Rahul Roy, Yves Sprumont and seminar participants at various conferences in which this paper was presented.
			Bhardwaj acknowledges financial support under PRIN project 20222Z3CR7  “Nudging under Limited Attention”. Chatterjee acknowledges financial support from the UKRI Frontier Research grant with grant number	EP/Z001528/1.	All errors are our own.
		}
	\end{center}
	
	\setcounter{footnote}{0}
	\def\thefootnote{\alph{footnote}}
	
	\vspace*{.1cm}
	\begin{center}
		\begin{tabular}{ccc}
			{\sc Bhavook Bhardwaj} & $\quad$ & {\sc Siddharth Chatterjee} \\
			{\small Universit\`{a} Ca' Foscari Venezia} & &  {\small University of Essex} \\
			{\small\tt bhavookb21@gmail.com} & & {\small\tt 123sidch@gmail.com}\\
		\end{tabular}
	\end{center}
	\vspace*{0.2cm}
	\begin{center}
		June 2026
	\end{center}
	\vspace*{0.4cm}
	
	{\small
		\noindent {\bf Abstract.} 	
		We develop a framework to study situations where decision makers face alternatives sequentially. Within this framework, we focus on endogenous stopping behavior using two broad classes of decision rules: \textit{stopping rules} and \textit{bounded stopping rules}. We establish the equivalence of these two classes and examine two of its implications. First, focusing on the procedural aspects of decision making, we  define \textit{computable} rules using the model of a Turing machine. Our equivalence result enables us to show that computable rules are implementable by finite automata. Second, we extend the setup of abstract choice theory beyond choice from sets and finite lists, to that from \textit{infinite sequences} of alternatives. The equivalence result allows us to derive \textit{testable implications} of choice behavior. We develop a revealed-preference ``toolkit'' and use it to characterize a threshold-based and a satisficing choice procedure.

		\medskip\par\noindent
		{\bf Keywords:\/} decision rules, bounded rationality, sequences, computability
		
		\medskip\par\noindent
		{\bf JEL Classification Numbers:\/} D01, D09
	}
	
	\setcounter{footnote}{0}
	\def\thefootnote{\arabic{footnote}}
	
	\bigskip
	
	\section{Introduction}\label{sec:introduction}

	To capture situations of decision making where the ordering of alternatives can affect final decisions,  \cite{rubinstein2006model} enriched the classical model of abstract choice and introduced choice functions over \textit{lists} which are ordered sets. However, there are many situations where the decision maker (DM) has to also decide on when to stop scanning alternatives as they keep on being presented. Examples of such situations include processing information, receiving recommendations, job search, meeting people etc. Such situations require \textit{endogenous} stopping in addition to making a decision and are not fully captured by the setup of finite lists\textemdash the last entry of the list is an exogenous stopping point.

	To model the above described situations, we introduce a general framework of decision making where a DM--represented by a \textit{decision rule}--processes infinite sequences of alternatives. In order to capture the notion of endogenous stopping, we require the DM to have a stopping ``point'' for every conceivable sequence and the decision is made by looking only at the string of alternatives that appear before the stopping point. That is, decisions are non-anticipatory-- contingent only upon the history of alternatives observed until the stopping point. Decision rules with this property are called \textit{stopping rules}. This is a broad class of rules which can accommodate, among other things, ``satisficing" behavior in this framework (see Section \ref{sec:testability}). 
	
	In the standard notion of a stopping rule, a DM may not stop at a finite point for some sequences. For instance, consider a job searcher who is receiving i.i.d. wage offers \`{a} la \cite{mccall1970economics}. The searcher's behavior in the optimal strategy requires stopping for a sequence only if it contains a wage offer above a reservation/cutoff wage; otherwise, the searcher does not stop and gets some outside option such as an ``unemployment benefit" in each period. This allows for a zero measure set of sequences on which the stopping time is infinite. However, we additionally require the DM to stop for \textit{every} sequence. While this may seem restrictive at a first glance, it is a natural requirement in many situations as we briefly outline next.

	First, the ``value" of the outside option may be low. In such cases, continuing forever may not be rational due to the implicit tradeoffs. Second, there may be implied costs associated which blow up as some underlying state evolves over time. Third, in a model with uncertainty, the DM may have an imperfect understanding of the environment which results in deteriorating beliefs over time. Fourth, our requirement of stopping for every sequence translates to the decisiveness of a decision rule when alternatives are examined sequentially over time or space. While some of the above suggested channels may induce stopping even with a Bayesian decision maker, other attitudes towards uncertainty such as ambiguity aversion or regret minimization would a fortiori induce stopping on all sequences.  
	
	Within the class of stopping rules, we consider a further subclass of decision rules called \textit{bounded stopping rules}. These rules require a ``global" finite bound on the stopping points. In other words, there exists a point such that for every sequence, the decision is made \textit{within} that point. This requirement is stricter than the requirement of stopping rules where we only require a ``pointwise" bound. These rules are of interest due to at least two possible reasons.
	
	First, finite attention that is indicative of cognitive limitations of a human being and finite processing capacity that is indicative of computational constraints of a machine would force stopping before a certain point irrespective of the sequence. The common bound required for bounded stopping rules can be seen as an analogue of a ``consideration set" in our framework. That is, the DM does not consider anything beyond the common bound irrespective of the ``menu" (a sequence). This is in line with the large literature on bounded rationality in choice that studies how these endogenous constraints affect choice behavior. Second, specific features of the situation that is modeled can induce endogenously a finite bound on stopping. For instance, consider the price search model of \cite{9b4ebbb8-234a-34c9-9343-c596e15cbb9a} which is formally similar to the job search model of \cite{mccall1970economics}. However, an important distinction is that of imperfect knowledge on part of the DM of the underlying distribution of the prices (as against wages). This distinction in the form of a ``second-order" uncertainty forces optimal search to terminate in finite time and therefore is a bounded stopping rule. 
	
	While every bounded stopping rule is a stopping rule by definition, our first result shows that the converse is also true. That is, the two classes are equivalent. We call this the ``Reduction Lemma''.  This stands in contrast to the lists framework of \cite{rubinstein2006model} and we can show that if we allow for lists of arbitrary size, this equivalence breaks down. In this paper, we study two implications of this result.
	First, we examine an implication of this result on computational and procedural aspects of decision making. Computational considerations are a recurring theme in economic theory. In our setup, we study these aspects using the model of a Turing Machine. A Turing machine is an abstract model of computation that embodies the idea of a procedural description of decision rules. We call a decision rule \textit{computable} if it can be implemented by a Turing machine. As it turns out, stopping rules and computable rules are equivalent. Next, we look at a model of computation that has been used in economic theory and game theory to model aspects of bounded rationality. This is the model of a finite automaton. Following the definition of computable decision rules, we define automaton-implementable rules analogously. In any ``finitary'' setup, decision rules implementable by Turing machine coincide with decision rules implementable by finite automata. However, this is not true if we allow for lists of arbitrary size in the framework introduced by \cite{rubinstein2006model}.  We show that due to the Reduction Lemma, computable rules coincide with finite automaton implementable rules in our setup.

	The second implication of our result concerns \textit{testability} and revealed preference aspects of choice heuristics. As is the case with certain theories of decision making in an ``infinite'' setup, axioms characterizing the theory may not be testable. A case in point is the Expected Utility theory where the axiom of \textit{continuity} is not testable experimentally. Similar problems can arise in our setup when one aims to empirically refute certain procedures of choice within the subclass of stopping rules. However, the Reduction Lemma ensures that all the axioms characterizing those stopping rules are testable by effectively making ours a ``finite'' setup.
	
	As highlighted above, stopping rules are a broad class of rules that allow for studying various types of behaviors. We introduce two natural choice procedures within the class of stopping rules: a threshold based choice rule and a satisficing rule. In order to develop revealed preference tests for choice behavior, a key challenge is to provide a ``language" in which axioms or testable conditions can be stated. To that end, we develop a revealed preference ``toolkit". This is essentially a set of concepts using which we state axioms on stopping rules. To illustrate, one of the concepts is that of ``sufficiency" of a finite segment. For a given decision rule, we say that a finite segment is \textit{sufficient} if the choice is the same for all ``extensions" of this segment. That is, any sequence that contains the same initial segment produces the same choice. Similarly, we define the notion of "minimal sufficiency" if it is sufficient and no further "truncation" of it is sufficient. These notions are closely related to the ones used in statistics. Using these notions, we develop our revealed preference toolkit and demonstrate its applicability by axiomatically characterizing the two choice rules.

	The layout of the paper is as follows. Section \ref{sec:setup} introduces the setup and the Reduction Lemma is stated and proved. Section \ref{sec:comp} introduces computable and finite automaton implementable rules and shows their equivalence in our setup and non-equivalence in the lists setup. Section \ref{sec:choice} develops the revealed preference toolkit followed by a characterization of two choice procedures in Section \ref{sec:testability}. Section \ref{sec: conclusion} provides a discussion of the related literature and concludes.

	\section{Setup and The Reduction Lemma}\label{sec:setup}
	Let $X$ be a non-empty finite set of alternatives. A \textit{sequence} is a map $S : \mathbb{N} \to X$, where $\mathbb{N}$ denotes the set of natural numbers. By $X^{\mathbb{N}}$ we denote the collection of all $X$-valued sequences. That is, $ X^{\mathbb{N}} \coloneqq \{S \ | \ S : \mathbb{N} \to X\}$. The term $S(i)$ corresponds to the $i^{th}$ entry of the sequence $S$. A \textit{segment} is any map $M: [k] \to X$, where $[k] \coloneqq \{1, \ldots, k\}$ for some $k \in \mathbb{N}$. Let the set of all segments of length $k$ be denoted by $\mathcal{S}_k$ and the set of all segments be denoted by $\mathcal{S}$. Consider any subset $E$ of natural numbers and a sequence $S$. We define the \textit{restriction} of $S$ to $E$ as the map $S|_E : E \to X$ where $[S|_E](i) = S(i)$ for all $i \in E$. When $E=[k]$ for some $k \in \mathbb{N}$, the segment $S|_{[k]}$ is called the \textit{truncation} of $S$ at $k$. We will abuse notation and write $S|_k$ instead of $S|_{[k]}$ whenever no confusion arises. For any $S,T \in X^{\mathbb{N}}$ and $k \in \mathbb{N}$, we define the \textit{concatenation} of the segment $S|_k$ and the sequence $T$ to be the sequence $S|_k\cdot T \in X^{\mathbb{N}}$ such that $[S|_k \cdot T](i) = S(i)$ for all $ i \in [k]$ and $ [S|_k \cdot T](i) = T(i-k)$ for all $i \in \{k+1, \ldots \}$. Concatenation of two segments is defined in a similar manner.
	We denote the set of all \textit{decisions} by a non-empty set $Y$. In particular, $Y$ can be equal to $X$. The DM in our model is represented by a decision rule, $d$, which gives a unique decision for every infinite sequence. Formally, it is defined as follows.
	\begin{defn}
		A decision rule on sequences is any map $d: X^{\mathbb{N}} \to Y$.
	\end{defn}
	
	A decision rule is more general than a ``choice'' rule as we do not restrict the decision to be a part of the input sequence i.e. we do not require $d(S)=S(i)$ for some $i \in \mathbb{N}$ (we put this additional requirement when we study choice behavior in Section \ref{sec:choice}). To illustrate its generality, suppose $X = \{0,1\}$ and $Y = \{\text{TRUE}, \text{FALSE}\}$. Consider a DM that is a computer program receiving bitstreams that represent expressions in a natural language (for instance, English) encoded in binary expression i.e. $0's$ and $1's$. For every input bitstream, the program declares it as ``TRUE" if it contains a grammatically correct sentence. It outputs ``FALSE" otherwise. This is a valid decision rule but would not involve stopping for all sequences. In what follows next, we focus on decision rules that capture the notion of endogenous stopping.
	
	\subsection{Stopping and bounded-stopping rules}
	
	Stopping rules capture the idea that for any given sequence, the DM does not wait indefinitely and ``makes up its mind" by a finite amount of time i.e. after viewing a finite initial segment and the subsequent alternatives of the sequence do not affect the decision.
	Formally, they are defined as follows.
	
	\begin{defn}
		A decision rule $d$ is a stopping rule if for all $ S \in X^{\mathbb{N}}$, there exists a $k \in \mathbb{N}$ such that for all $T \in X^{\mathbb{N}}$ with $T|_k = S|_k$,  we have $d(S) = d(T) $.
	\end{defn}
	
	To show that not every decision rule is a stopping rule, consider the following simple example. Let $X=\{x^*,y\}$. The decision rule $d$ is defined as $d(S)= x^*$ if $x^* = S(i)$ for some $i \in \mathbb{N}$ and $d(S)=S(1)$ otherwise. Consider any sequence $S$ that does not feature $x^*$ in it. It can be observed that for any $k \in \mathbb{N}$, we can find a $T \in X^{\mathbb{N}}$ that features $x^*$ in its $k+1^{\text{th}}$ entry and therefore there does not exist a $k$ for $S$ as required in the definition of a stopping rule. If our interpretation of a decision rule is that the sequence is examined by the DM sequentially \textemdash in discrete time for instance \textemdash then such a decision rule looks implausible since for the sequences that do not feature $x^*$, the DM will never stop and would have to wait ``forever" to make a decision.

	It is important to note the stopping point or the ``relevant" finite segment for stopping rules can depend on the sequence. Since the set of sequences is infinite, the lengths of these relevant segments are not guaranteed to have a finite upper bound. A subclass of stopping rules for which these lengths have a finite upper bound are called bounded stopping rules. 
	
	\begin{defn}
		A decision rule $d$ is a bounded stopping rule if there exists a $k \in \mathbb{N}$ such that for all $S,T \in X^{\mathbb{N}}$, if $S|_k = T|_k$, then we have $d(S) = d(T)$.
	\end{defn}
	
	While stopping rules require for every sequence, the existence of a finite bound on the ``consideration" of the DM, bounded stopping rules require a fixed finite bound on the consideration for \textit{every} sequence. That is, there is a change in the order of quantifiers in the definition of the two subclasses of decision rules. The following is a simple example of a bounded stopping rule: The DM is endowed with a preference order $ \succ$ over $X$, and for any sequence, she considers only the first 10 alternatives if the first element of the sequence is some designated $x^* \in X$ and picks the $ \succ$-maximal alternative from them. Otherwise, she looks at the first 20 alternatives and picks the $ \succ$-maximal alternative from them.   
	

	\subsection{The Reduction Lemma}
	
	Our main result establishes the equivalence of stopping and bounded stopping rules. Before stating and proving the result, we first provide an alternative definition of stopping rules using the following useful object which is defined for any decision rule $d$.
	$$	k_d(S)\coloneqq\inf\big\{k\in\mathbb{N}:d(S)=d(S\vert_k\cdot T)\text{ for all }T\in X^{\mathbb{N}} \big\}$$
	The function $k_d(\cdot)$ is the \textit{stopping time} for the sequence $S$ and captures the smallest truncation of a sequence $S$ beyond which the terms of the sequence do not affect decisions. Using $k_d$, we redefine stopping and bounded stopping rules as follows (with the convention that $n < \infty$ for all $n \in \mathbb{N}$).

	\begin{defn}
		A decision rule $d$ is a
		\begin{itemize}
			\item[(i)] Stopping rule if $k_d(S)<\infty$ for every $S \in X^{\mathbb{N}}$.  
			\vspace{-0.1cm}
			\item[(ii)]Bounded stopping rule if $\sup\{k_d(S) : S \in X^{\mathbb{N}}\}<\infty$.
		\end{itemize} 	
	\end{defn}
	
	While it is clear by the definition above that every bounded stopping rule is a stopping rule, we now show that the converse is also true. 
	
	\begin{theorem}
		Every stopping rule is a bounded stopping rule.
	\end{theorem}
	\begin{proof}
		Let $d:X^{\mathbb{N}}\rightarrow X$ be a stopping rule. Suppose, for the sake of contradiction, $d$ is not a bounded stopping rule. The proof is organized in three steps.
		\vspace{0.2cm}
		
		\indent\textit{Step 1}: We iteratively define a sequence of pairs $\{(k_j,\mathcal{A}_j)\}_{j\in\mathbb{N}}$, where $k_j\in\mathbb{N}$ and $\mathcal{A}_j\subseteq X^{\mathbb{N}}$, as follows:
		\begin{enumerate}
			\item	Let $k_1\coloneqq\inf\{k_d(S):S\in X^{\mathbb{N}}\}$ and $\mathcal{A}_1\coloneqq\{S\in X^{\mathbb{N}}:k_d(S)=k_1\}$.
			\vspace{-0.1cm}
			
			\item	For any $j\in\mathbb{N}\setminus\{1\}$, assuming $(k_l,\mathcal{A}_l)$ have already been defined for every $l\in\{1,\hdots,j-1\}$, let
			\begin{align*}
				k_j &\coloneqq\inf\{k_d(S):S\in X^{\mathbb{N}}\setminus\cup_{l=1}^{j-1}\mathcal{A}_l\}\text{, and }	\\
				\mathcal{A}_j &\coloneqq\{S\in X^{\mathbb{N}}\setminus\cup_{l=1}^{j-1}\mathcal{A}_l:k_d(S)=k_j\}
			\end{align*}
		\end{enumerate}
		The sets $\mathcal{A}_j$ refer to the set of all the sequences (henceforth inputs\footnote{ Since our proof involves constructing sequences of sequences of alternatives and subsequences of those sequences, we use the term ``input" to denote a sequence of alternatives to avoid any confusion.}) for which the stopping time is $k_j$.
		\noindent From our supposition that $d$ is stopping rule and $d$ does not have a finite bound on the set of stopping times, the following properties are immediate:
		\begin{enumerate}[(a)]
			\item	For each $j\in\mathbb{N}$, $k_j\in\mathbb{N}$ and $\mathcal{A}_j\neq\varnothing$.
			\item	$k_1<k_2<\hdots<k_j<\hdots$ and so on. Further, $k_i \geq i$ for all $ i \in \mathbb{N}$.
			\item	$\{\mathcal{A}_j:j\in\mathbb{N}\}$ is a partition of $X^{\mathbb{N}}$.
		\end{enumerate}
		These properties shall be referred to in the rest of the argument.
		
		\begin{figure}
			\begin{center}	
				\begin{tikzpicture}
					\draw[thick] (7,-0.5) --  (1,-0.5) -- (1,0) --  (7,0);
					\draw[thick] (1.4,-0.6)--(1.4,0.1)--(2.1,0.1)--(2.1,-0.6)--(1.4,-0.6);
					\node[thick] at (1.75,0.3) {\footnotesize $k_1$};
					\foreach \i in {1.5,2,2.5,3,3.5,4,4.5,5,5.5,6}
					{\draw[thick] (\i,-0.5) -- (\i,0);} 
					\node[thick] at (6.5,-0.25) {$\cdots$};
					\node[thick] at (0.7,-0.25) {$S_1$};
					\draw[thick] (7,-1.5) -- (1,-1.5) -- (1,-1) --  (7,-1);
					\draw[thick] (2.9,-1.6)--(2.9,-0.9)--(3.6,-0.9)--(3.6,-1.6)--(2.9,-1.6);
					\node[thick] at (3.25,-0.7) {\footnotesize $k_2$};	
					\node[thick] at (3.9,-1.9) { $\ddots$};
					\node[thick] at (4.35,-2.2) { $\ddots$};
					\foreach \i in {1.5,2,2.5,3,3.5,4,4.5,5,5.5,6}
					{\draw[thick] (\i,-1.5) -- (\i,-1);} 
					\node[thick] at (6.5,-1.25) {$\cdots$};
					\node[thick] at (0.7,-1.25) {$S_2$};
					\node[thick] at (0.7,-1.8) {$\vdots$};
					\node[thick] at (0.7,-2.2) {$\vdots$};
					\node[thick] at (0.7,-2.6) {$\vdots$};	
					\draw[thick] (7,-3.5) -- (1,-3.5) -- (1,-3) --  (7,-3);
					\draw[thick] (4.9,-3.6)--(4.9,-2.9)--(5.6,-2.9)--(5.6,-3.6)--(4.9,-3.6);
					\node[thick] at (5.25,-2.7) {\footnotesize $k_j$};	
					\node[thick] at (6.1,-3.8) { $\ddots$};
					\node[thick] at (6.55,-4.1) { $\ddots$};
					\foreach \i in {1.5,2,2.5,3,3.5,4,4.5,5,5.5,6}
					{\draw[thick] (\i,-3.5) -- (\i,-3);} 
					\node[thick] at (6.5,-3.25) {$\cdots$};
					\node[thick] at (0.7,-3.25) {$S_j$};
					\node[thick] at (0.7,-3.8) {$\vdots$};
					\node[thick] at (0.7,-4.2) {$\vdots$};
					\node[thick] at (0.7,-4.6) {$\vdots$};	 	 
				\end{tikzpicture}
				\caption{A sequence of inputs $\{S_i\}_{i \in \mathbb{N}}$ with increasing stopping times}
			\end{center}
		\end{figure}
		
		\indent\textit{Step 2}: For every $j \in \mathbb{N}$, pick an arbitrary $S_j\in\mathcal{A}_j$. This generates a sequence $\{S_i\}_{i \in \mathbb{N}}$ of inputs such that the stopping time for each $S_j$ is $k_j$. By property (b), we know that this corresponds to an increasing sequence of stopping times. 
		Now, we construct a subsequence $\{S^*_i\}_{i \in \mathbb{N}}$ of the above sequence with the following progressive ``agreement" property: For all $k \in \mathbb{N}$, we have 
		$S_k^*|_{k} = S_j^*|_{k}$ for all $j \geq k$. To do this, we use the following lemma.
		\begin{lemma}
			For any $k \in \mathbb{N}$ and a sequence of inputs $\{T_{i}\}_{i \in \mathbb{N}}$ where $T_i \in X^{\mathbb{N}}$, there exists a subsequence $\{T_{k_i}\}_{i \in \mathbb{N}}$ such that $T_{k_m}|_k = T_{k_n}|_k$ for all $m,n \in \mathbb{N}$.
		\end{lemma}
		\begin{proof}
			Consider any $k \in \mathbb{N}$. Since $X$ is finite, the number of possible segments of length $k$ is $|X|^k$. Since  $\{T_{i}\}_{i \in \mathbb{N}}$ is an infinite collection of inputs, by the pigeonhole principle, there exists at least one segment of length $k$, say $M$, that is repeated infinitely often and  therefore we can construct a subsequence $\{T_{k_i}\}_{i \in \mathbb{N}}$ , such that $T_{k_i}|_k = M$ for all $i \in \mathbb{N}$. 
		\end{proof}
		
		Now using the above lemma, we recursively define an indexed collection of sequences of inputs $\{\{S_{k_i}\}_{i \in \mathbb{N}}\}_{k \in \mathbb{N}}$ as follows: 
		\begin{itemize}
			\item For $k=1$ and $\{S_i\}_{i\in \mathbb{N}}$, applying Lemma 1 we get a subsequence $\{S_{1_i}\}_{i \in \mathbb{N}}$ such that $S_{1_i}|_1 = S_{1_j}|_1$ for all $i,j \in \mathbb{N}$. 
			\item For $k \geq 2$, applying Lemma 1 on the sequence $\{S_{(k-1)_i}\}^{\infty}_{i=2}$, we get a subsequence $\{S_{k_i}\}_{i \in \mathbb{N}}$ such that $S_{k_i}|_{k}= S_{k_j}|_{k}$ for all $i, j \in \mathbb{N}$.
		\end{itemize}
		
		\begin{figure}
			\begin{center}
				\begin{tikzpicture}
					\draw[thick] (7,-0.5) --  (1,-0.5) -- (1,0) --  (7,0);
					\draw[thick] (1.9,-0.6)--(1.9,0.1)--(2.6,0.1)--(2.6,-0.6)--(1.9,-0.6);
					\draw[thick,cyan] (1.05,-5) --(1.05,0.05) --(1.45,0.05) -- (1.45,-5);
					\node[thick] at (2.25,0.30) {\tiny $k_d(S^*_1)$};
					\foreach \i in {1.5,2,2.5,3,3.5,4,4.5,5,5.5,6}
					{\draw[thick] (\i,-0.5) -- (\i,0);} 
					\node[thick] at (6.5,-0.25) {$\cdots$};
					\node[thick] at (0.7,-0.25) {$S^*_1$};
					\node[thick] at (1.3, -0.30) {\footnotesize $x$};
					\draw[thick] (7,-1.5) -- (1,-1.5) -- (1,-1) --  (7,-1);
					\draw[thick] (2.9,-1.6)--(2.9,-0.9)--(3.6,-0.9)--(3.6,-1.6)--(2.9,-1.6);
					\draw[thick,cyan] (1.55,-5) --(1.55,-1.05) --(1.95,-1.05) -- (1.95,-5);
					\node[thick] at (3.25,-0.7) {\tiny $k_d(S^*_2)$};	
					\node[thick] at (3.9,-1.9) { $\ddots$};
					\node[thick] at (4.40,-2.25) { $\ddots$};
					\foreach \i in {1.5,2,2.5,3,3.5,4,4.5,5,5.5,6}
					{\draw[thick] (\i,-1.5) -- (\i,-1);} 
					\node[thick] at (6.5,-1.25) {$\cdots$};
					\node[thick] at (0.7,-1.25) {$S^*_2$};
					\node[thick] at (1.3, -1.3) {\footnotesize $x$};
					\node[thick] at (1.75, -1.3) {\footnotesize $y$};
					\node[thick] at (0.7,-1.8) {$\vdots$};
					\node[thick] at (0.7,-2.2) {$\vdots$};
					\node[thick] at (0.7,-2.6) {$\vdots$};	
					\draw[thick] (7,-3.5) -- (1,-3.5) -- (1,-3) --  (7,-3);
					\draw[thick] (4.9,-3.6)--(4.9,-2.9)--(5.6,-2.9)--(5.6,-3.6)--(4.9,-3.6);
					\draw[thick,cyan] (2.55,-5) --(2.55,-2.05) --(2.95,-2.05) -- (2.95,-5);
					\node[thick] at (2.3,-1.75){$\ddots$};
					\node[thick] at (3.25,-2.35){$\ddots$};
					\draw[thick,cyan] (3.55,-5) --(3.55,-3.05) --(3.95,-3.05) -- (3.95,-5);
					\node[thick] at (5.25,-2.7) {\tiny $k_d(S^*_j)$};	
					\node[thick] at (6,-3.8) { $\ddots$};
					\node[thick] at (6.50,-4.15) { $\ddots$};
					\foreach \i in {1.5,2,2.5,3,3.5,4,4.5,5,5.5,6}
					{\draw[thick] (\i,-3.5) -- (\i,-3);} 
					\node[thick] at (6.5,-3.25) {$\cdots$};
					\node[thick] at (0.7,-3.25) {$S^*_j$};
					\node[thick] at (1.3, -3.3) {\footnotesize $x$};
					\node[thick] at (1.75, -3.3) {\footnotesize $y$};
					\node[thick] at (2.25, -3.3) {\footnotesize $z$};
					\node[thick] at (2.75, -2.3) {\footnotesize $w$};
					\node[thick] at (2.75, -3.3) {\footnotesize $w$};
					\node[thick] at (0.7,-3.8) {$\vdots$};
					\node[thick] at (0.7,-4.2) {$\vdots$};
					\node[thick] at (0.7,-4.6) {$\vdots$};	 
					
					\draw[thick,cyan] (4.55,-5) --(4.55,-4.05) --(4.95,-4.05) -- (4.95,-5);	
					\draw[thick] (7,-6) -- (1,-6) -- (1,-5.5) --  (7,-5.5);
					\node[thick] at (4.75,-4.25){\footnotesize $a$};
					\foreach \i in {1.5,2,2.5,3,3.5,4,4.5,5,5.5,6}
					{\draw[thick] (\i,-6) -- (\i,-5.5);} 
					\node[thick] at (6.5,-5.75) {$\cdots$};
					\node[thick] at (0.7,-5.75) {$S^*$};
					\node[thick] at (1.3, -5.8) {\footnotesize $x$};
					\node[thick] at (1.75, -5.8) {\footnotesize $y$};
					\node[thick] at (2.25, -5.8) {\footnotesize $z$};
					\node[thick] at (2.75, -5.8) {\footnotesize $w$};
					\node[thick] at (3.25, -5.8) {\footnotesize $\cdot$};
					\node[thick] at (3.75, -5.8) {\footnotesize $\cdot$};
					\node[thick] at (4.25, -5.8) {\footnotesize $\cdot$};			
					\node[thick] at (4.75,-5.8){\footnotesize $a$};			 
				\end{tikzpicture}
				\vspace{0.1cm}
				\caption{Progressive agreement of $\{S^*_i\}_{i \in \mathbb{N}}$ and the target input $S^*$}
			\end{center}	
		\end{figure}	
		
		Starting from $\{S_i\}_{i \in \mathbb{N}}$ in stage 0, at every stage $k \geq 1$, we generate a sequence of inputs such that all the inputs in the sequence have the same initial $k$-long segment. Note that this indexed collection of sequences is nested i.e. $\{S_{(k+1)_i}\}_{i \in \mathbb{N}}$ is a subsequence of  $\{S_{k_i}\}_{i \in \mathbb{N}}$ for all $k \in \mathbb{N}$. Therefore,  for any $k,l \in \mathbb{N}$ such that $ l>k$, we have $S_{k_i}|_k= S_{l_j}|_k$ for all $i,j \in \mathbb{N}$. In particular $S_{k_1}|_k= S_{l_1}|_k$. Now we define the required sequence of inputs as  follows: $S^*_k \coloneqq S_{k_1}$ for all $k \in \mathbb{N}$. That is, $S^*_k$ is equal to the first term (input) of the sequence generated at the $k^{th}$ recursion of the above definition. The sequence of inputs $\{S^*_i\}_{i \in \mathbb{N}}$ thus generated has the property that $S^*_k|_k = S^*_j|_k$ for all $j \geq k$. Further, $\{S^*_i\}_{i \in \mathbb{N}}$ is a subsequence of $\{S_i\}_{i \in \mathbb{N}}$ and hence corresponds to an increasing sequence of stopping times $k_d(S^*_1) < k_d(S^*_2) \ldots < k_d(S^*_j) < \ldots $ where $k_d(S^*_i) \geq k_i \geq i$.   Finally, we define the input $S^*\in X^{\mathbb{N}}$ as $$S^*(i) \coloneqq S^*_i(i) \ \quad \ \ \  \forall i \in \mathbb{N}$$
		
		It can be observed that due to the progressive ``agreement", the sequence of inputs $\{S^*_i\}_{i \in \mathbb{N}}$ ``converges" to the input $S^*$ i.e. for any $k \in \mathbb{N}$, $S^*|_k = S^*_j|_k$ for all $j \geq k$. 
		
		\vspace{0.2cm}
		
		\indent\textit{Step 3:} Since $d$ is a stopping rule, there must exist $k^* \in \mathbb{N}$ such that 
		\begin{equation}
			d(S^*) = d(S^*|_{k^
				*}\cdot T) \ \ \forall \ \ T \in X^{\mathbb{N}}
		\end{equation}
		
		Consider any $l > k^*$. Note that since $k_d(S^*_1) < k_d(S^*_2) \ldots < k_d(S^*_j) < \ldots$, there exists $S^*_l \in \{S^*_i\}_{i \in \mathbb{N}}$ such that $k_d(S^*_l) > k^*$. By the definition of $S^*$ and the progressive agreement property of $\{S^*_i\}_{i\in \mathbb{N}}$, we know that $S^*|_k = S^*_l|_k $ for all $k \leq l$. In particular $S^*|_{k^*} = S^*_l|_{k^*} $ and therefore we can write $S^*_l$ as the concatenation of  $S^*|_{k^*}$ and $T' \in X^{\mathbb{N}}$, where $T'(i) =  S^*_l(k^* + i)$ for all $i \in \mathbb{N}$. But then by (1), we have 
		
		$$d(S^*_l) = d([S^*|_{k^*}] \cdot T' ) = d([S^*|_{k^*}] \cdot T) \ \ \ \ \forall T \in X^\mathbb{N}$$
		which implies that $k^*$ is the stopping time for $S^*_l$, a contradiction. Therefore our initial supposition that $d$ is not a bounded stopping rule is wrong implying that $d$ must be a bounded stopping rule.		
	\end{proof}
	
	\subsection{Some remarks}
	
	Our proof relies on a diagonalization argument to construct the sequence $S^*$ and there are two critical ingredients in the proof. First, the finiteness of $X$ enables us to establish Lemma 1 and second, the assumption of full domain, $X^\mathbb{N}$, allows us to construct the target sequence $S^*$ which leads to the final contradiction. There exists a topological approach to our result in which we can show that for any stopping rule $d$, the function $k_d : X^{\mathbb{N}} \to \mathbb{N}$ is continuous when $X^{\mathbb{N}}$ and $\mathbb{N}$ are endowed with the product topology and the discrete topology respectively. Using the Tychonoff Theorem, we can then observe that $X^{\mathbb{N}}$ is compact and hence the function $k_d$ is uniformly continuous. The uniform continuity of $k_d$ gives us the finite bound on the values of $k_d$.


	The equivalence of stopping and bounded stopping rules highlights the fact that with the assumption of endogenous stopping, we end up showing that only a finite number of segments are ``relevant" in decision making. While this finiteness provides appropriate grounds for testability of various decision procedures (see Section \ref{sec:choice}), it also provides surprising results when we explore computational aspects of decision making in our setup.  This is the content of the next section.
	
	\section{Computability and Bounded Rationality}\label{sec:comp}

	
	It is widely accepted that cognitive limitations and computational constraints have an important role in the decision making process.  While the standard notion of rationality that is synonymous with unrestricted maximization assumes no such constraints, there are a variety of settings where the assumption of infinite processing capabilities is an unrealistic one. As \cite{richter1999computable} remark 
	\begin{quote} 
		\begin{spacing}{1}
			``Can human beings really work with arbitrarily complex preferences, utility	functions, and technologies as classical economic theory assumes? \ldots real people, using ‘realistic’ languages, cannot communicate arbitrary real-number quantities and prices.
			Realism, then, suggests that we restrict ourselves to ‘simple’ preferences, utility functions, and technologies..."
		\end{spacing}
	\end{quote}
	
	\vspace{-0.5cm}
	
	In order to incorporate computational constraints in our model, a natural first question to ask is that what decision rules are \textit{computable}? To answer this question, we turn to an abstract model of computation: the \textit{Turing machine}. A physical description of a Turing machine involves two objects: A finite state machine and an infinite ``tape" which enables it to have an effectively ``infinite memory". According to the Church-Turing thesis, any physically realizable computer can be represented using a Turing machine and this makes it the most powerful model of computation known till date. In other words, the question of computational feasibility of a decision problem can be thought of as equivalent to that of its Turing-implementability. Therefore, we call a decision rule computable if it can be implemented by a Turing machine.
	
	\begin{figure}[ht]
		\begin{center}
			\begin{tikzpicture}[every node/.style={block},
				block/.style={minimum height=1.5em,outer sep=0pt,draw,rectangle,node distance=0pt}]
				\node (A) {$a$};
				\node (B) [left=of A] {$\ b \ $};
				\node (C) [left=of B] {$ \ b \ $};
				\node (D) [right=of A] {$\quad$};
				\node (E) [right=of D] {$\quad $};
				\node (F) [above = 0.75cm of A, draw=black,thick, scale=1.5] {\tiny Finite Control};
				\node (A') [above = 0.75cm of F]{$\ 1 \ $};
				\node (B') [left=of A'] {$\ 0 \ $};
				\node (C') [left=of B'] {$ \ 1 \ $};
				\node (D') [right=of A'] {$ \quad  $};
				\node (E') [right=of D'] {$\quad $};
				\draw[-latex] (F) -- (A);
				\draw[-latex] (F) -- (A');
				\draw (C.north west) -- ++(-1cm,0) (C.south west) -- ++ (-1cm,0) 
				(E.north east) -- ++(1cm,0) (E.south east) -- ++ (1cm,0);
				\draw (C'.north west) -- ++(-1cm,0) (C'.south west) -- ++ (-1cm,0) 
				(E'.north east) -- ++(1cm,0) (E'.south east) -- ++ (1cm,0);
			\end{tikzpicture}
			\caption{A two-tape Turing machine}
		\end{center}
	\end{figure}
	
	What we mean by implementing a decision rule is that there exists a Turing machine such that for any input (a sequence) that is fed (formally defined below) into it, the machine halts and produces the same output as the decision rule on that input.   Before describing the decision making process in our setup using a Turing machine, we first provide its formal definition.
	
	\begin{defn}
		A Turing machine  is a tuple $ TM= ( Q, X, \delta, O)$, where $Q$ is a finite set of states, $X$ is a finite set of symbols (alternatives), $ \delta : Q \times X^2 \to Q \times X \times \{L,S,R\}^2$ is a transition function and $ O : \mathcal{S} \to Y$ is an output function.\footnote{Where $\mathcal{S}$ and $Y$ are as defined in Section \ref{sec:setup}.}
	\end{defn}

	We conceptualize the DM as a Turing machine with a finite number of states denoted by $Q$ and two \textit{tapes} \textemdash input and output/working tape\textemdash which are infinite one directional line of ``cells". Each tape is equipped with a tape head. The tape head of the input tape reads the symbols on the tape one cell at a time whereas the tape head of the output tape can write or rewrite symbols to the tape one cell at a time.

	In the standard setup, the inputs to a Turing machine are finite strings from a finite ``alphabet" ($X$). The input in our setup, an infinite sequence, is written on the input tape preceded by a symbol $\triangleleft \notin X$ and the decision-making process is as follows: The symbol $\triangleleft$ \textit{initializes} the machine and it begins in some initial state $q_0$. Then, it ``parses" through the input one at a time using the transition function $\delta$. Depending on the current state and the entries under the two tape heads, the transition function determines the next state, the movement of the tape heads (left, right or stay) and the entry on the output tape. There is a designated set of \textit{terminal} states and once the machine enters a terminal state, it \textit{halts}. The decision is then made using the output function, $O$, using the segment generated on the output/working tape.  Using this notion of a Turing machine, we are now equipped to state a formal definition of computable decision rules.
	
	\begin{defn}
		A decision rule $d$ is computable, if there exists a Turing machine $TM_d$ such that for all $S \in X^{\mathbb{N}}$, (i) The Turing machine halts; and (ii) $TM_d(S) = d(S)$. 
	\end{defn}
	
	It is easy to see that not every decision rule is a computable rule. In particular, it is worth observing that rationality \textemdash defined as preference maximization \textemdash is incompatible with computability i.e. rational choice rules are not computable. To show this, consider a preference order $\succ$ over $X$ and consider any sequence $S \in X^\mathbb{N}$ such that it does not feature the $\succ$-maximal element in it. Then, no Turing machine will halt for this input. This is in line with \cite{kramer1967impossibility} who shows that when the DM suffers from computational constraints, it is impossible to display fully rational behavior.

	Having defined computability of decision rules, we now introduce another model of computation which is simpler and has been widely used to model various aspects of bounded rationality: a finite automaton.  In the context of repeated games, automata have been used to incorporate the cost of complexity of strategies (see \cite{rubinstein1986finite}) and in the context of individual decision-making, they have been used to describe the procedural aspects of decision-making (\cite{salant2011procedural}). It is formally defined as follows. 
	
	\begin{defn}
		An automaton is a tuple $A = (Q, X, \delta, O )$ where $Q$ is a finite set of states, $X$ is a finite set of symbols (alternatives), $\delta: Q \times X \to Q $ is a transition function and $O: F \to Y$ is an output function where $F \subset Q$ is the set of terminal states.   	
	\end{defn}
	
	The DM is conceptualized as an automaton in a similar way as a Turing machine. The input is written on an input tape.  It starts in an initial state $q_0 \in Q$ and reads elements of an input one at a time. However, an important difference is that the tape head can move only in one direction.
	For every input element and the current state, the transition function determines the next state of the automaton and  the tape head moves to the next element. Within the set of states is a designated set of \textit{terminal} states, denoted by $F$. Once the automaton enters one of these states, it halts. An output function $O$ then produces a decision based on the terminal state.\footnote{In \cite{salant2011procedural}, the output function requires both the state as well as the symbol under the tape head to produce the output. Such machines are called Mealy machines. Whereas our formulation is similar to the one in \cite{rubinstein1986finite} and such machines are called Moore machines.}

	In the classical theory of computation, a finite automaton is a simpler model since it does not have an infinite tape to simulate an effectively infinite memory. The following example illustrates this difference: Let $X =\{a,b, \varnothing\}$, $Y=\{\textbf{yes},\textbf{no}\}$ and $\mathcal{S}^o$ be the set of all finite segments that comprise of $a$ and $b$ with the last element being $\varnothing$ (indicating the end of the segment). The decision rule outputs ``\textbf{yes}'' for any $S \in \{a^nb^n\varnothing: n \in \mathbb{N}\} $ i.e. any segment that comprises of $n$ number of $a'$s followed by $n$ number of $b'$s for any $n \in \mathbb{N}$ and it outputs ``\textbf{no}'' otherwise. Such a decision rule can be implemented using a Turing machine. However, it cannot be implemented by a finite automaton. This example also illustrates the fact that automaton-implementable rules over \textit{lists} form a strict subclass of computable rules.
	
	\begin{figure}
		\begin{center}
			\begin{tikzpicture} [node distance = 2.6cm, on grid, auto]
				
				\node (q0) [state, initial, accepting] {$q_o$};
				\node (q1) [state, above = of q0] {$1_x$};
				\node (q2) [state, below = of q0] {$1_y$};
				\node (q3) [state, right = of q1] {$2_x$};
				\node (q4) [state, right = of q2] {$2_y$};  
				\node (q5) [state, below = of q3] {$1_x1_y$ }; 
				\path [-stealth, thick]
				(q0) edge  node {$x$} (q1)
				(q0) edge  node [left] {$y$} (q2)
				(q1) edge node[above right] {$x$} (q3)
				(q5) edge node[above right] {$x$} (q3)
				(q5) edge node[above right] {$y$} (q4)
				(q2) edge node[below right] {$x$} (q5)
				(q2) edge node[below right] {$y$} (q4)
				(q1) edge node[above right] {$y$} (q5)
				(q3) edge[loop right] node{$x,y$} (q3)
				(q4) edge[loop right] node{$x,y$} (q4);
			\end{tikzpicture}
			\caption{An automaton-implementable decision rule}
			\label{figure1}		
		\end{center}
	\end{figure}
	
	\vspace{0.2cm}
	
	Analogous to the previous definition of computable rules, we call decision rules \textit{finite automaton-implementable} or simply \textit{automaton-implementable} if they can be implemented by a finite automaton. Formally, they are defined as follows. 
	
	\begin{defn}
		A decision rule $d$ is automaton-implementable if there exists a finite automaton $A_d$ such that for all inputs $S \in X^{\mathbb{N}}$, (i) The automaton halts and; (ii) $A_d(S) =d(S)$.
	\end{defn}
	
	To illustrate automaton-implementable decision rules, consider the following example. Suppose $X=\{x,y\}$, each alternative has a ``weight'' of 1. The DM has a ``threshold'' of 2 and uses the following procedure: for every sequence of alternatives, she selects the first alternative whose cumulative weight (due to repetitions) crosses the threshold weight. Then this decision rule can be implemented using an automaton with 5 states, excluding the initial state, $q_0$ (see Figure 4). Our notions of computable and automaton-implementable rules are closely linked to stopping rules and bounded stopping rules and using the Reduction Lemma, we now show that all of them are in fact equivalent.

	\begin{theorem}
		Every computable rule is automaton-implementable.
	\end{theorem}
	\begin{proof}
		Let $d$ be a computable decision rule implementable by a Turing machine $TM_d$. Consider any arbitrary input $S \in X^{\mathbb{N}}$. Since the Turing machine halts for $S$, there exists $k \in \mathbb{N}$ such that $TM_d$ does not examine alternatives in $S$ beyond $S|_k$. Consider any $S' \in X^{\mathbb{N}}$ such that $S'|_k= S|_k$. The Turing machine does not examine alternatives beyond $k$ in $S'$ as well and since $S'|_k= S|_k$, we have $d(S')=d(S)$. Therefore $d$ is a stopping rule. By the Reduction Lemma (Theorem 1), it is a bounded stopping rule. So, there exists $k^* \in \mathbb{N}$ such that for all $S,T \in X^{\mathbb{N}}$, we have $d(S) = d([S|_{k^*}] \cdot T)$. An automaton with at most $\sum_{j=1}^{k^*-1}|X|^{j}$ non-terminal states (one for each segment of length less than $k^*$) and $|X|^{k^*}$ terminal states (one for each segment of length $k^*$) can implement this decision rule. Therefore, it is automaton-implementable.
	\end{proof}
	
	\noindent \textbf{Remark.} The construction of the automaton in the above proof is the most ``inefficient" in terms of the state complexity i.e. the number of states. This is the largest number of states required to implement a bounded stopping rule/simply computable rule. To illustrate this fact, consider the example above. Going by the construction given in the proof, the automaton will have 6 states instead of 5, excluding the initial state (the additional state being $1_y1_x$ which will be different from $1_x1_y$).

	\section{Choice and Revealed Preference}\label{sec:choice}
	One of the central objects of study in abstract choice theory is a choice function. The domain of a choice function is a collection of ``menus" and for every menu it gives the choice set, namely the set of chosen or ``choosable" alternatives \textit{from} that menu.\footnote{Multi-valued functions are often called choice correspondences.}  In the classical theory, these menus correspond to \textit{sets} of alternatives. The analogue of a menu in our model is an infinite sequence and that of a single-valued choice function is what we term as a choice rule. These form a subclass of decision rules that require the DM to choose an alternative from within the sequence. They are formally defined as follows.   
	
	\begin{defn}
		A choice rule $d$ is a map	$d: X^{\mathbb{N}} \to X$ such that for all $S \in X^{\mathbb{N}}$, $d(S)=S(i)$ for some $i \in \mathbb{N}$.
	\end{defn}
	
	If we restrict our attention to choice rules that are also stopping rules, we get the further restriction that for any sequence $S$, the choice must lie within the initial $k_d(S)$-long segment. An import of Theorem 2 is that stopping rules are equivalent to computable decision rules. While studying choice rules, we maintain that the assumption of computability is a \textit{normative} one. This is in line with the interpretation that in our model, the DM encounters alternatives sequentially, one at a time. Therefore the assumption of stopping or equivalently that of computability is a plausible one. In this section, we first provide a characterization of computable choice rules and then provide a ``toolkit" for conducting revealed preference analysis in our setup. 
	
	\subsection{Continuity and computability}

	With the added structure to the decision space in the case of choice rules, we provide a characterization of computable choice rules via their continuity with respect to a natural topological structure on the domain and the co-domain. This is captured in the following result, a proof of which is relegated to the appendix.
	
	\begin{theorem}
		Consider a non-constant choice rule $d$ and assume $X$ and $X^{\mathbb{N}}$ are endowed with the discrete and the product topology respectively. Then $d$ is computable if and only if it is continuous. 
	\end{theorem}
	
	The domain of inputs, $X^{\mathbb{N}}$, is often referred to as a Cantor space and the product topology on it as Cantor topology. In the proof of the result, we first show the structure of basic open sets in the product topology, which are referred to as open cylinders. These correspond to all inputs which ``agree" on one location. Finite intersections of such open cylinders are called cylinder sets and they form the basis for this topology. Therefore the behavioral interpretation is that a DM considers two sequences ``approximately" same by comparing only finitely many initial locations. Continuity of the choice rule then implies that the DM cannot display ``jumps" for close enough choice problems.

	When restricting our attention to choice rules, the reduction lemma allows us to succinctly represent computable choice rules via finite \textit{trees}. To show this, we revisit the example in the previous section again with $X = \{x,y\}$, both the alternatives having a weight of 1 and the threshold weight being 2. The root node of the tree is the ``null" symbol and every path from the root node to a terminal node corresponds to a $k_d(\cdot)$ long segment. Further, in the case of the procedure in this example, the terminal node indicates the choice from the $k_d(\cdot)$ segment. This provides a complete description of the choice rule (see Figure 5). 
	
	\begin{figure}[ht]
		\begin{center}
			\begin{tikzpicture}[
				every node/.style = {minimum width = 2em, draw, circle, scale=0.8, thick},
				scale=1]
				\node {$\varnothing$} [sibling distance= 2cm]
				child {node [xshift= -.5cm]{x}
					child{node [fill=lightgray]{x}}
					child{node [xshift= -0.5cm]{y}   child {node [fill=lightgray][xshift= -0.2cm]{x}} child {node [fill=lightgray][xshift= -0.8cm]{y}}}}
				child {node [xshift= 0.5cm]{y}
					child{node [xshift= 0.2cm] {x}  child {node [fill=lightgray] [xshift= 1.2cm]{x}} child {node [fill=lightgray] [xshift= 0.3cm]{y}}}
					child{node [fill=lightgray] [xshift= 0.2cm] {y}}};
			\end{tikzpicture}	
		\end{center}
		\caption{Tree representation of a computable choice rule}
	\end{figure}
	
	\subsection{A revealed preference toolkit}
	In the theory of choice from sets, the analysis of different choice procedures involves imposing some consistency properties\textemdash called axioms\textemdash on choice functions. These axioms are often in the form of ``contraction" or ``expansion" properties i.e. consistency of choices across menus that are related via set inclusion. In order to conduct an axiomatic analysis of choice procedures in our setup, we require a suitably adapted ``language" to state such axioms on choice rules. To that end, we introduce two useful informational concepts of \textit{sufficiency} and \textit{minimal sufficiency} of segments. In order to define these formally, we require some notation. Recall that a segment is any map $M : [k] \to X$ where $k \in \mathbb{N}$ and the set of all segments is denoted by $\mathcal{S}$. Denote domain of a segment $M$ as dom($M$). We define a strict partial order\footnote{An asymmetric and transitive binary relation.} over the set of all segments $\mathcal{S}$ as follows: for any $M,M' \in \mathcal{S}$, let $M \triangleright M'$ if and only if (i) dom($M'$) $\subsetneq$ dom($M$) and (ii) $M(i) = M'(i)$ for all $ i \in$ dom($M'$). The relation $\triangleright$ is thus the ``extending" relation and $M \triangleright M'$ can be interpreted as the segment $M$ ``extends" the segment $M'$. A sufficient segment is defined as follows.
	
	\begin{defn}
		For a decision rule $d$, a segment $M$ is \textit{sufficient} if 
		$d(M \cdot T)=d(M \cdot T')$ for all $T , T' \in X^{\mathbb{N}}$.
	\end{defn}
	
	The intuitive content of the above definition is as follows. As the DM faces a sequence $S$, there comes a point $k\in\mathbb{N}$ when the segment $M = S|_k$ has enough information for the decision maker to have ``made up its mind" i.e. $M$ is informationally ``sufficient" to enforce a decision. However, the acquired information will \textit{not} be sufficient until a certain point. This motivates the notion of \textit{minimal sufficiency}. 
	
	\begin{defn}
		For a decision rule $d$, a segment $M$ is \textit{minimal sufficient} if it is \textit{sufficient} and for any $M'$ such that $M \triangleright M'$, the segment $M'$ is not \textit{sufficient}.
	\end{defn}
	
	Minimal sufficiency captures the idea of the ``critical" length of a segment to enforce a decision. By critical, we mean that if the  segment is smaller than that length, it can no longer guarantee the same decision for all concatenated sequences. Note that the definition of stopping rules indicates that every sequence must have a corresponding \textit{minimal sufficient} segment that ``implements" the decision. For a given stopping rule, $d$,  let the class of all \textit{sufficient} and \textit{minimal sufficient} segments be denoted by $\mathscr{S}$ and $\mathscr{MS}$ respectively. If $M = S|_k$ for some $k \in \mathbb{N}$ and $M$ is a sufficient segment for a decision rule $d$, then we will abuse notation and denote the decision for $S$ by $d(M)$ i.e. $d(S)=d(M)$.
	
	To illustrate the idea of sufficiency and minimal sufficiency, let us revisit the decision procedure in the example of the previous section. Let $X = \{a,b,c\}$ and all alternatives have weight 1. Suppose the DM has a threshold value of 3 and consider the sequence $S= (a \ b \ c \ a \ b \ c ..)$ i.e. it consists of ``cycles" of alternatives $a$, $b$ and $c$. Here, the minimal sufficient segment is of length 7 i.e. where $a$ is the first alternative to appear 3 times. Any initial segment of $S$ with length less than 7 is not minimal sufficient and any segment with length more than 7 is sufficient.

	Since the set of sequences is infinite, in practice, the identification of sufficient (and minimal sufficient) segments is not possible for a given stopping rule. In order to overcome this problem of identification, we introduce a new object \textemdash a \textit{decision procedure} \textemdash that captures the ``dynamic" aspect of decision-making. Denote by $\star$, a symbol not in the decision space $Y$, representing ``indecision". Decision procedures are maps that are defined on the set of all finite segments. They map every finite segment to either ``indecision" or some decision in $Y$.
	
	\begin{defn}
		A decision procedure is any map $d_* : \mathcal{S} \to Y \cup \{\star\}$ such that 
		\vspace{-0.2cm}
		\begin{enumerate}
			\item[(i)] If $M' \triangleright M$ and $d_*(M') =\star$, then $d_*(M) = \star$
			\vspace{-0.15cm}
			\item[(ii)]  If $M' \triangleright M$ and $d_*(M)  \in Y$, then  $d_*(M) = d_*(M')$
			\vspace{-0.15cm}
			\item[(iii)] For any sequence of segments $M_1, M_2, \ldots$ satisfying $M_{k+1} \triangleright M_k $ for all $k \in \mathbb{N}$, there exists $k \in \mathbb{N}$ such that $d_*(M_k) \in Y$
		\end{enumerate}
	\end{defn}
	
	A decision procedure can be thought of as a dynamic representation of a stopping rule. If a DM is represented by a decision procedure, the three consistency requirements can be interpreted in the following manner. First, if the DM has not made a decision at a given point (in time or space), i.e. at a segment, then she would have not made a decision at any preceding point as well i.e. at any sub-segment of that segment. Second, if the DM has made a decision at a given point, then for any subsequent point as well she makes a decision. Further, she makes the same decision at any subsequent point as well. Third, for any progressively increasing sequence of segments, she makes a decision at \textit{some} point along the sequence.

	The connection between stopping rules and decision procedures is made precise by defining a map that outputs a decision procedure for every stopping rule. Let $\mathcal{D}_s$ and $\mathcal{D}_*$ be the set of all stopping rules and decision procedures respectively. To each stopping rule $d$, associate the corresponding map $d_*^{\dagger} : \mathcal{S} \to Y \cup \{\star\}$ as follows:
	\begin{enumerate}
		\item  For any segment $M \in \mathcal{S}$, let $d_*^{\dagger}(M) \coloneqq \star$ if there exists $S \in X^\mathbb{N}$ and $k \in \mathbb{N}$ such that $k < k_d(S)$ and $M = S|_k$.
		\vspace{-0.2cm}
		\item For any segment $M \in \mathcal{S}$, let $d_*^{\dagger}(M) \coloneqq d(S)$ if there exists $S \in X^{\mathbb{N}}$ such that $k \geq k_d(S)$ and $M= S|_k$.
	\end{enumerate}
	
	\begin{lemma}
		The map $d_*^{\dagger}$ is well defined and $d_*^{\dagger} \in \mathcal{D}_*$.
	\end{lemma}
	\begin{proof}
		To show $d_*^{\dagger}$ is well defined, consider any arbitrary $M \in \mathcal{S}$. Suppose there exists $S \in  X^{\mathbb{N}}$ and $k \in \mathbb{N}$ such that $k < k_d(S)$ and $M = S|_{k}$. Then by definition $d_*^{\dagger}(M) = \star$. Assume for contradiction that there exists $S' \in X^{\mathbb{N}}$, $S' \neq S$ such that $M = S'|_{k}$ and $k \geq k_d(S')$. Now, since $S'|_{k_d(S')} = S|_{k_d(S')}$, by the definition of a stopping rule, we must have $k_d(S') = k_d(S)$, a contradiction. Now, suppose there exist $S,S' \in X^{\mathbb{N}}$ and $k \in \mathbb{N}$ such that $S|_k=S'|_k =M$ and $\max(k_d(S), k_d(S')) \leq k$. Since $S|_k=S'|_k=M$, by the definition of stopping rules we must have $k_d(S) = k_d(S')$ and hence $d(S)=d(S') =d(M)$. Therefore, $d_*^{\dagger}$ is well-defined. \\
		Now, to show $d_*^{\dagger}$ is a decision procedure, consider an arbitrary $M \in \mathcal{S}$. If $d_*^{\dagger}(M) = \star$, then consider any $M'$ such that $M \triangleright M'$. Since there exists $S \in X^{\mathbb{N}}$ such that $k < k_d(S)$ and $M = S|_k$, we know that $M' = S|_{k'}$ for some $k' < k$ and hence we have $d_*^{\dagger}(M') = \star$. Now, suppose  $d_*^{\dagger}(M) = d(S)$ for some $S \in X^{\mathbb{N}}$. We know that $M = S|_{k}$ for some $k$ and $k \geq k_d(S)$. Consider any $M' \in \mathcal{S}$ such that $M' \triangleright M$. For any $S' \in X^{\mathbb{N}}$ and $k' >k$ such that $S'|_{k'} = M$, we know that $S'|_{k_d(S)} = S|_{k_d(S)}$ Therefore we must have $k_d(S') = k_d(S)$ and hence $d(S)=d(S') = d_*^{\dagger}(M')$. Finally consider any sequence of segments $M_1, M_2, \ldots$  such that $M_{k+1} \triangleright M_{k}$ for all $k \in \mathbb{N}$. Now, since every stopping rule is a bounded stopping rule, there exists a $S \in X^{\mathbb{N}}$ and $k \in \mathbb{N}$ such that $k \geq k_d(S)$ such that $M_{k'} = S|_k$ for some $k' \in \mathbb{N}$.  Therefore, we have $d_*^{\dagger}(M_{k'}) = d(S) \in Y$.     
	\end{proof}
	Using the above lemma, we let $ \eta : \mathcal{D}_s \to \mathcal{D}_*$ be defined as: 
	$$ \eta(d) \coloneqq d_*^{\dagger} \ \ \text{for every} \ d \in \mathcal{D}_S $$
	
	This map provides a natural way to assign a unique decision procedure for every stopping rule. Further, for every decision procedure there exists a unique stopping rule. This claim is established by the following result. 
	
	\begin{proposition}
		The map $\eta :  \mathcal{D}_s \to \mathcal{D}_*$ is a bijection. 
	\end{proposition}
	\begin{proof}
		To show that $\eta$ is one-to-one, consider $d,d' \in \mathcal{D}_s $  such that $d \neq d'$. Therefore there exists $S \in X^{\mathbb{N}}$ such that $d(S) \neq d'(S)$. Let $k \coloneqq \max(k_d(S), k_{d'}(S))$ and consider $M = S|_k$. Let $\eta(d) = d_*^{\dagger}$ and $ \eta(d') = d_*^{' \dagger}$. By definition $ d_*^{\dagger}(M)=d(S)\neq d'(S)= d_*^{' \dagger}(M)$. Therefore, $\eta$ is one-to-one. \\
		To show that $\eta$ is onto, consider any arbitrary $d_* \in \mathcal{D_*}$. We need to define $d \in \mathcal{D}_s$ such that $\eta(d) = d_*$. Define $d: X^{\mathbb{N}} \to Y$ as follows. Consider any $S \in X^{\mathbb{N}}$ and the sequence of segments $S|_1, S|_2, \ldots$. Note that $S|_{k+1} \triangleright S|_k$ for all $k \in \mathbb{N}$ and by definition of a decision procedure, there exists $k \in \mathbb{N}$ such that $d_*(S|_k) \in Y$. Let $k^* = \inf \{k \in \mathbb{N}: d_*(S|_k) \in Y\}$. Define $d(S) \coloneqq d_*(S|_{k^*}) \in Y$. By definition of $d_*$, we have $d(S|_{k^*} \cdot T) =d(S)$ for all $T \in X^{\mathbb{N}}$ and $ k_d(S) = k^*$. Therefore $d$ is a stopping rule. Also, by the definition of $\eta$, we have $\eta(d)= d_*$. Therefore, $\eta$ is onto. 
	\end{proof}
	
	We have shown that there is a natural bijection between the class of stopping rules and that of decision procedures. While mathematically equivalent, decision procedures and stopping rules are conceptually different objects. Stopping rules process entire infinite sequences whereas decision procedures show how the DM processes information when the infinite sequences are presented ``gradually" in a dynamic manner. To study physical settings, decision procedures provide a more realistic model of a DM. The minimal sufficiency and sufficiency of segments are naturally defined for decision procedures as follows: For a decision procedure $d_*$, a segment $M$ is	
	\begin{enumerate} 
		\item[(i)] \textit{Minimal sufficient} if $d_*(M) \in Y$ and $d_*(M') = \star$ for all $M'$ such that $M \triangleright M'$. 
		\item[(ii)] \textit{Sufficient} if it is minimal sufficient or there exists $M' \in \mathcal{S}$ such that $M'$ is minimal sufficient and $M \triangleright M'$. 
	\end{enumerate}	
	
	Decision procedures are useful from the revealed preference perspective as they enable complete identification of minimal sufficient and sufficient segments in finitely many steps. For instance, given a tree representation of a decision procedure that corresponds to a choice rule, a Depth First Search (DFS) algorithm will output the class of all minimal sufficient segments \textemdash and consequently sufficient segments \textemdash in finite time.   As we will see in the next section, these fully identifiable segments will help us formulate axioms to behaviorally characterize some natural choice procedures.

	\section{Choice procedures and testability}\label{sec:testability}
	In this section, we operationalize the revealed preference toolkit developed in the previous section. We do that by providing axiomatic characterization of two choice procedures.\footnote{The choice procedures formulated in this section form a subclass of stopping rules. Since stopping rules are equivalent to decision procedures, these can be formulated as decision procedures as well and all the results go through. However, for expositional and notational convenience, we will operate in the domain of stopping rules.} 
	The first one is a formalization of the example introduced in Section \ref{sec:comp} of ``weights'' and a ``threshold''.  Here is a motivating example.
	
	\textbf{Example 1.} Let $X$ be a finite set of movies and $\{X_i\}^N_{i=1}$ denote a partition of the set of movies into $N$ ``genres" for some $N \in \mathbb{N}$. A DM wishes to watch a movie and relies on recommendations. She attaches a ``weight" to each genre which indicates the value she attaches to each genre i.e. there exists a function $w: N \to \mathbb{R}_+$ such that every movie in the $i^{th}$ genre is given the same weight. Her decision procedure involves seeking recommendations \textit{sequentially} from different sources such as peer groups, websites etc.  
	She has a ``threshold" weight in her mind and for every sequence of recommendations, she selects the first movie whose cumulative weight (due to repetitions) crosses the threshold weight.  
	
	The DM is equipped with two objects. The first one is a \textit{weight} function $w : X \to \mathbb{R}_{++}$  that assigns a positive real number to every alternative. The weights can be thought of as some scores the DM assigns to the alternatives that are indicative of the relative importance of alternatives. For instance, while seeking movie recommendations, a DM may give a higher score to ``action" movies over the ones belonging to the genre ``drama". The second object is a \textit{threshold} weight $v \in \mathbb{R}_+$. This  threshold corresponds to the satisficing component that the DM uses to make decisions. 
	
	The DM uses the following procedure to make choices. For any sequence, she ``parses" through it sequentially, maintaining a count of the cumulative weight of each alternative in a ``register". As soon as she encounters an alternative whose cumulative weight crosses the threshold, she stops and selects it. We call this procedure \textbf{Threshold based stopping} (TBS). In order to formally define this procedure, denote for any given sequence $S$ and a position $N \in \mathbb{N}$, the cumulative weight of an alternative $x$ as
	$$W^N_S(x) \coloneqq |\{i \in [N]: S(i) = x\}| \cdot w(x) $$
	
	Now, we can define TBS formally as follows.
	
	\begin{defn}
		A computable choice rule $d$ is a Threshold based stopping (TBS) rule if there exists $v \in \mathbb{R}_+$ and $w : X \to \mathbb{R}_{++}$ such that for any $S \in X^{\mathbb{N}} $,
		$$ d(S) = x^*(S)$$
		where $x^*(S) \in X$ is the unique alternative satisfying the following condition: $ W^N_S(x) \geq v > W^N_S(y) $
		for all $y \neq x$ and some $N \in \mathbb{N}$.
	\end{defn}
	
	Note that since the weights assigned to alternatives are positive real numbers, for any sequence, there exists some position in it such that exactly one alternative's cumulative weight crosses (weakly) the threshold at that position. This defines the stopping condition of the DM. This procedure is behaviorally characterized by two axioms. In order to state the first axiom, we introduce the concept of a \textit{favorable transformation} of a sequence with respect to an alternative. Intuitively, this involves bringing an alternative ``closer" to the DM by transforming that sequence into a new one. That is, by lowering its position, an alternative is examined earlier than it was examined previously. There are two ways to favorably transform a sequence with respect to an alternative. The first way is to interchange the position of that alternative with another alternative that precedes it in the input. Formally, for any sequence $S $ and $k\in\mathbb{N}$, let $\hat{S}^k\in X^{\mathbb{N}}$ be the sequence which is defined as follows:
	
	\[
	\hat{S}^k(i)\coloneqq
	\begin{cases}
		S(k+1)	&\text{ if }i=k;\\
		S(k)		&\text{ if }i=k+1;\\
		S(i)		&\text{ otherwise.}
	\end{cases}
	\]
	That is, the sequence $\hat{S}^k$ is obtained from $S$ by interchanging its $k^{th}$ and $(k+1)^{st}$ elements. We call $\hat{S}^k$ a \textit{favorable shift} of $S$ with respect to an alternative $x$ if $S(k+1)=x$. Denote the class of all favorable shifts of $S$ with respect to an alternative $x$ by $\mathcal{FS}(S,x)$. The second way to bring an alternative closer to the DM is by deleting another alternative. Formally, for any sequence $S$ and  $k \in \mathbb{N}$ let $\tilde{S}^k \in X^{\mathbb{N}}$ be the sequence defined as 
	\[
	\tilde{S}^k(i)\coloneqq
	\begin{cases}
		S(i)	&\text{ if }i<k;\\
		S(i+1)	&\text{ if }i \geq k
	\end{cases}
	\]		
	The sequence $\tilde{S}^k$ is obtained from $S$ by dropping the alternative located at the $k^{th}$ position. We call $\tilde{S}^k$ as a \textit{favorable deletion}  of $S$ with respect to an alternative $x$ if $S(k) \neq x$. Denote the class of all favorable deletions of $S$ with respect to an alternative $x$ be denoted by $\mathcal{FD}(S,x)$. 
	
	For any $S\in\mathcal{S}$ and $x\in X$, a \textit{favorable transformation} of $S$ with respect to $x$ is a  favorable shift \textit{or} a favorable deletion. The class of all favorable transformations of $S$ with respect to $x$ shall be denoted by $\mathcal{F}(S,x)$. Therefore, $\mathcal{F}(S,x)=\mathcal{FS}(S,x)\cup\mathcal{FD}(S,x)$ by definition. Our first condition requires that the stopping rule should be ``monotone" when it comes to favorable transformations with respect to the chosen alternatives. That is, it requires the DM to make the same choice if the chosen alternative is brought ``closer" to him via a favorable transformation.
	
	\textsc{Monotonicity}: A decision rule satisfies monotonicity if for any $S,S' \in X^{\mathbb{N}}$ such that $S' \in \mathcal{F}(S,x)$, $$[d(S)=x] \implies [d(S')=x] $$ 
	
	The second condition is about the effect on choice when a sufficient segment is concatenated to any truncation of a minimal sufficient segment. It states that if a minimal sufficient segment $M$ ``implements" an alternative $x$ and another sufficient segment $N$ that does not contain $x$ implements some other alternative, then concatenating any truncation of $M$ with $N$ prevents $x$ from being chosen.	In other words, it asserts that a sufficient segment not containing an alternative can ``dominate" a non-minimal sufficient segment in an informational sense. We say that for an alternative $x$ and a segment $M$, $x \notin M$ if $ M(i) \neq x$ for all $ i \in$ dom($M$). That is, $x \notin M$ when $x$ does not appear in any position of the segment $M$.
	
	\textsc{Informational Dominance}:  A decision rule $d$ satisfies informational dominance if for any $M \in \mathscr{MS} $ and $N \in \mathscr{S}$ such that $d(M)=x$ and $d(N) \neq x$ and for any $M'$ such that $M \triangleright M'$, 
	$$[x \notin N] \implies [x \neq d([M'\cdot N])]$$
	
	Note that since the segment $[M' \cdot N]$ contains a sufficient segment within it, it must be a sufficient segment itself. To illustrate this condition, consider a DM that assigns weight 1 to each alternative in $X = \{a,b,c\}$ and has a threshold weight of 3. The segment $M=(a \ b \ c \ a \ b \ c \ a )$ is a minimal sufficient segment with $d(M) =a$. Consider another segment $ N= (b \ b \ c \ b \ b )$. Note that this is a sufficient segment since $b$ appears 3 times in it. Now consider an arbitrary truncation $M'$ of $M$. Informational dominance requires that for any concatenation of $M'$ with the segment $N$, the choice cannot be equal to $a$. In this case, for any truncation $M'$, we can see that $d(M' \cdot N) =b$ since $b$ is the first alternative whose cumulative weight reaches 3.  We now show that these two conditions characterize cardinal satisficing behavior.
	
	\begin{theorem}
		A computable choice rule is a TBS if and only if it satisfies Monotonicity and Informational Dominance.
	\end{theorem}

	Now we turn to the second choice procedure which is satisficing heuristic. Satisficing, first introduced by Herbert Simon (see \cite{simon1955behavioral}) is an influential model of decision-making and has been studied widely in the literature (see \cite{kovach2020satisficing}, \cite{aguiar2016satisficing},  \cite{tyson2015satisficing} and \cite{papi2012satisficing}, among others).
	The basic idea underlying satisficing is that due to factors like computational constraints, complexity of the choice problem etc., a DM may not resort to optimizing behavior. Instead, based on a binary classification of the alternatives into acceptable/satisfactory and non-acceptable/unsatisfactory, she may select any alternative belonging to the former category. Satisficing behavior is often modeled as a DM examining alternatives sequentially until a ``good enough'' alternative is observed. While some existing models endogenize the search order of the DM (see \cite{aguiar2016satisficing}), others treat it as observable in the form of a list and vary the threshold (see \cite{kovach2020satisficing}).	Our setup provides a natural way to study satisficing behavior. Consider a DM who is represented by three objects: (i) A ranking over the set of alternatives $X$, denoted by $\succ$ which is a preference order,  (ii) A threshold alternative $a^* \in X$ which is used for the binary classification of the set of alternatives into satisfactory and unsatisfactory; and (iii) an attention parameter $k \in \mathbb{N}$ that specifies the relevant segment for any sequence.  
	
	The DM uses the following procedure to make a choice. For any sequence, she ``parses" through it sequentially. She stops if she encounters $a^*$ or an alternative that is ranked above $a^*$. Otherwise she stops after observing the first $k$ alternatives and chooses the $\succ$-maximal one from the set of observed alternatives\footnote{For any set $A$, the $\succ$-maximal set, denoted by $\max(A, \succ)$ is defined as $\max(A, \succ) \coloneqq \{x \in A| \neg y \succ x \ \forall \ y \in A \setminus \{x\}\}$.}.   This is in contrast with the satisficing model discussed in \cite{rubinstein2012lecture} where a DM chooses the last alternative from the list if it contains no alternative ranked above the threshold. To illustrate this procedure, consider an example where $X = \{a,b,c,d\}$ and the DM's preferences are $a \succ b \succ c \succ d$, the attention parameter $k$ is 2 and the threshold alternative is $b$. Consider a sequence $S = (c \ d \ a \ a \dots)$. Since the first two positions do not contain any satisfactory alternative, the choice is $c$ whereas the choice from the sequence $S' = (a \ b \ c \ c \ldots )$ is $a$. 
	
	In order to formally define our procedure, we let for any alternative $x \in X$ its upper and lower contour set with respect to $ \succ$ be denoted by $U(x)$ and $L(x)$, respectively. That is, $U(x) \coloneqq \{y \in X: y \succ x\}$\footnote{Note that since $\succ$ is assumed to be reflexive, $x \in U(x)$ for all $x \in X$.} and $L(x) \coloneqq X \setminus U(x)$. Now, we formally define this procedure.
	
	\begin{defn}
		A computable choice rule $d$ is an Ordinal Satisficing Rule (OSR) if there exists $(\succ , a^*, k)$  such that for any $S \in X^{\mathbb{N}}$, 
		\[ d(S) =
		\begin{cases}
			S(i) & \text{where} \ i \in [k], S(i) \in U(a^*) , S(j) \in L(a^*) \ \forall j < i ; \\
			\max( \bigcup_{i \in [k]} S(i), \succ) & \text{if} \ S(i) \in L(a^*) \ \forall \ i \in [k] 	
		\end{cases}
		\]
		
	\end{defn}

	Consider two extreme cases of this procedure. The first is when everything is satisfactory. That is, $x \succ a^*$ for all $x \in X$. In this case, the DM always chooses the first alternative that is presented to her. On the other hand, if $a^*$ is the $\succ$-maximal alternative, then this procedure corresponds to ``attention-constrained" \textit{rational} behavior. That is, within the limited attention span of the DM, she always chooses the best alternative. Therefore rational behavior is only a special case which is in contrast with satisficing over sets where it is indistinguishable from preference maximization irrespective of the threshold (See \cite{rubinstein2012lecture}).
	
	This procedure is behaviorally characterized by three axioms. Before we state these axioms, we need to define the concept of a \textit{decisive} alternative. The idea behind a decisive alternative is that whenever it is present in a minimal sufficient segment, it is chosen. Intuitively, it dominates attention of the DM and enforces its choice. For an arbitrary choice rule, the set of decisive alternatives may be empty. However, we will show that in the case of OSR, it turns out to be non-empty. Further, in the special case where the DM is an \textit{attention-constrained} preference maximizer i.e. she chooses the $\succ$-maximal alternative after viewing a fixed length of alternatives,  this set will be a singleton.
	
	\begin{defn}
		For a decision rule $d$, an alternative $x$ is decisive if for all $M \in \mathscr{MS}$,
		$$[x \in M] \implies [x = d(M)]$$	
	\end{defn}
	We say that an alternative is \textit{non-decisive} if it is not decisive. Denote by $D$ and $D'= X \setminus D$, the set of all decisive alternatives and non-decisive alternatives, respectively. Now,  based on the notion of a decisive alternative, let $\mathscr{MS}_{D'}$ and $\mathscr{MS}_{D}$ denote the classes of all minimal sufficient segments that do not contain a decisive alternative and that do contain a decisive alternative, respectively. That is, $\mathscr{MS}_{D'}  \coloneqq  \{M \in \mathscr{MS} :  M \subseteq D'\}$ and $\mathscr{MS}_{D}  \coloneqq  \mathscr{MS} \setminus \mathscr{MS}_{D'}$. 
	
	\vspace{0.2cm}	 
	
	Our first condition is an adaptation of condition $\alpha$ (also called Chernoff's condition) for single valued choice functions to our setup. In the case of menus as sets, condition $\alpha$ requires that if an alternative that is chosen from a menu, say $A$ and it is present in a smaller menu, say $B$ where $B \subset A$, then it must be chosen from $B$ also. In our setup, we say that if an alternative is chosen in minimal sufficient segment and it is present in another minimal sufficient segment whose range is contained within the range of the former segment, then it must be chosen in the new segment as well.  Notice that this is true for a decisive alternative by definition, therefore we impose it for only minimal sufficient segments belonging to $\mathscr{MS}_{D'}$. For any two segments $M$ and $M'$, if $x \in M$ implies $x \in M'$, we abuse notation and write $M \subseteq M'$.
	
	\textsc{Sequential-$\alpha$}: 
	A choice rule $d$ satisfies sequential-$\alpha$ if for any $M, M' \in \mathscr{MS}_{D'}$ such that $M \subseteq M'$,
	$$[d(M') \in M] \implies [d(M')=d(M)]$$ 
	
	Our next axiom is an adaptation of the No Binary Cycles (NBC) condition on choice functions over sets. NBC requires that binary choices cannot display cycles implying that in the case of single valued choice functions, the pairwise revealed relation must be transitive. The analogue of a binary menu in our setup is a minimal segment that has only two alternatives in it. For any two distinct alternatives $x,y$, denote by $M_{xy}$ any segment that has only $x$ and $y$ in it. Therefore, our NBC condition prevents choices from these segments to display cycles.

	\textsc{Sequential-NBC}:  A choice rule $d$ satisfies sequential-NBC if for any $x,y,z \in X$ and $M_{xy},M_{yz}, M_{xz} \in \mathscr{MS}_{D'}$,  
	$$[d(M_{xy})=x,  d(M_{yz})=y] \implies [d(M_{xz})\neq z]$$ 
	
	Our final axiom is about the effect of replacing alternatives on the informational content of a minimal sufficient segment. It requires that if the occurrence of a non-decisive alternative in a minimal sufficient segment is replaced by some other non-decisive alternative, its informational content should remain the same. In other words, the new segment should still be a sufficient segment. Again, this applies to only segments that do not have decisive alternatives since it holds for decisive alternatives by definition. 
	
	\textsc{Replacement}:  A choice rule $d$ satisfies replacement if for any $x,y \in D'$ and $M,M'  \in \mathcal{S}_m$ such that $M(i)= x$, $M'(i)=y$ for some $i \in [m]$ and $M(j)=M'(j)$ for all $j \neq i$, $j \in [m]$
	$$[M \in \mathscr{MS}_{D'}] \implies [M' \in \mathscr{S}]$$
	
	To illustrate, consider the following choice procedure that \textit{violates} this condition: Let $X = \{x,y,z\}$ and $d$ be defined as follows
	
	$$ d(S) = \begin{cases}
		S(3)\ \ \text{if} \ S(3) \in \{x,y\} \\
		S(5)\ \ \text{otherwise}    
	\end{cases} $$ 
	
	For any sequence $S$, the DM looks at the third location. If it is either $x$ or $y$, then she chooses it, otherwise she picks the alternative at the fifth location. The minimal sufficient segments are of size 3 or 5. Further, no alternative is decisive. Now consider a segment $M = (x \ y \ x)$. It is clear that this segment is a minimal segment since for all $S \in X^{\mathbb{N}}$, we have $d([M \cdot S]) =x$. Now, create $M'$ by replacing $x$ in the third location with $z$. That is, $M' = (x \ y \ z)$. The segment $M'$ is no longer sufficient since for any sequence $S = (x \ y \ z \ x \ y \ \ldots)$, we have $d(S) =y$, whereas for any $S' = (x \ y \ z \ x \ z \ \ldots)$, we have $d(S')=z$.   Now, we are ready to state our result.
	
	\begin{theorem}
		A computable choice rule is an OSR if and only if it satisfies Sequential-$\alpha$, Sequential-NBC and Replacement. 
	\end{theorem}

	\section{Related Literature and Concluding Remarks}\label{sec: conclusion}
	
	The idea that a DM may observe alternatives in the form of a list, i.e. an ordered set, was first formalized in choice theory by \cite{rubinstein2006model}. Based on their framework, a variety of models incorporating order and framing effects in choice have been introduced (see for instance \cite{horan2010sequential}, \cite{guney2014theory} and \cite{dimitrov2016divide}); more broadly, \cite{salantrubinstein2008frames} develop a general theory of choice in the presence of \textit{frames}\textemdash observable but payoff-irrelevant information, of which the order in which alternatives are presented is a leading example. A related strand rationalizes order- and menu-dependent behavior through the sequential application of heuristics or through limited consideration, as in the sequentially rationalizable choice of \cite{manzini2007sequentially} and the revealed-attention model of \cite{masatlioglu2012revealed}.

	Satisficing, first introduced by \cite{simon1955behavioral}, has been an influential idea in the choice-theoretic literature and there have been several adaptations of it. The list setup provides a natural framework to study satisficing behavior. \cite{kovach2020satisficing} introduce one such model in which the DM searches through the list until she has seen \textit{k} alternatives and then chooses from those she has seen, while \cite{tyson2008cognitive} relates satisficing to failures of contraction consistency under cognitive constraints. Another adaptation, interpreted as one of \textit{approval} rather than choice, is given by \cite{manzini2019sequential}. The behavioral content of satisficing sharpens when the search order is itself observable: \cite{caplin2011searchsatisficing} provide experimental ``choice-process" evidence that subjects search sequentially and stop once a reservation level of utility is reached. The normative benchmark for such sequential search is the reservation-rule characterization of \cite{weitzman1979optimal}, itself rooted in the classical economics of information of \cite{stigler1961economics} and the job-search model of \cite{mccall1970economics}. Since our framework is a generalization of lists, satisficing heuristics are a natural object of study, and we introduced a satisficing procedure in Section \ref{sec:testability}.
	
	Computational aspects of decision-making have long been of interest to economic theorists. As pointed out by \cite{richter1999computable}, computability-based theories also provide foundations for complexity analysis, and the idea of bounded rationality has been closely linked to the computational limitations of an economic agent (see \cite{futia1977complexity}). This connection is most developed in the analysis of repeated games, where players' strategies are implemented by finite automata (\cite{rubinstein1986finite}) whose use is taken to be costly in the number of states (\cite{abreu1988structure}). Relatedly, \cite{gilboa1988complexity} studies the complexity of computing best responses among such automata, and \cite{kalai1988finite} link strategic complexity to the size of the automata used in repeated play. In the context of finitely repeated games, bounding the complexity of strategies is shown to justify cooperation (see \cite{neyman1985bounded}).

	The discussion of the role of computational constraints in individual decision-making goes back to \cite{simon1955behavioral}. He remarks ``\ldots limits on computational capacity may be important constraints entering into the definition of rational choice under particular circumstances". To capture the finite informational processing capacity, \cite{kramer1967impossibility} models the DM as a finite automaton and shows the incongruence of rational decision-making given this behavioral restriction on the DM. On the other hand, \cite{salant2003limited} shows that with the finite automaton model of decision-making, implementation of rational choice functions over menus (sets of alternatives) is computationally efficient, while implementation of other choice functions is much more computationally demanding. A conceptual difficulty underlying this program\textemdash that modeling a boundedly rational agent as \textit{optimally} choosing her decision procedure invites an infinite regress\textemdash is examined by \cite{lipman1991decide}.
	
	While modeling DMs using an automaton has been a popular approach to model aspects of bounded rationality, to model the computational aspects of decision-making, the Turing machine is a more appropriate device. It is a more powerful model of computation than a finite automaton and can be thought of as a precise way of mathematically describing an algorithm. Turing machines embody the idea of computability in the truest sense. \cite{richter1999computable} use the idea of a Turing machine to define computable preferences and show that computable preferences have computable utility representations.  \cite{camara2021computationally} models the DM as a Turing machine in the environment of decision-making under risk. He introduces the notion of computational tractability. A decision problem is \textit{intractable} if it cannot be implemented by an algorithm in a ``reasonable" amount of time. He shows that expected utility maximization is intractable unless the utility function satisfies a strong separability property.
	
	Our object of interest, infinite sequences, in the context of choices has been previously studied by  \cite{caplin2011search}. Our model differs from their model in terms of incorporating sequences in the domain of choice functions whereas they enrich the observable choice data by incorporating sequences as the output of the choice function and interpret these sequences as provisional choices of the DM with contemplation time.

	In this paper, we introduced a new model of decision making that considers infinite sequences as primitives as against sets or finite lists. This model provides a natural setting to study decision-making situations where the DM faces alternatives sequentially and also provides a generalization of the framework on choice over lists introduced in \cite{rubinstein2006model}. Further, our model allows a study of situations where the decision to stop examining alternatives is completely endogenous to the DM. To that end, we introduced a natural subclass of decision rules called stopping rules that require the DM to decide after viewing a finite segment of every sequence. Our main result \textemdash the Reduction Lemma \textemdash showed that the class of stopping rules is equivalent to its seemingly stricter subclass\textemdash that of bounded-stopping rules.  We introduced the notion of computability of a decision rule using Turing machines and showed that any computable decision rule can be implemented by a finite automaton \textemdash a result that does not hold in the setup of decisions over finite lists.

	The Reduction Lemma allows us to develop a language to formulate \textit{testable} conditions for studying different choice procedures. This involves defining the informational concepts of sufficiency and minimal sufficiency of finite segments in decision-making. With a dynamic representation of stopping rules in the form of decision procedures, we showed that these segments are completely identifiable in practice. To demonstrate the applicability of our model, we introduced some natural choice procedures and provided their behavioral characterizations. Future work will involve studying stochastic choice rules and examining potential applications of our main result.
	
	\pagebreak
	
	\section*{Appendix: Omitted Proofs}

	\begin{center}
		\textsc{Proof of Theorem 3}
	\end{center}
	
	First, we describe the structure of product topology on $X^{\mathbb{N}}$ where $X$ is endowed with the discrete topology. We know that $\Pi_{X^{\mathbb{N}}}$, the product topology, is the smallest topology with respect to which the projection maps are continuous. Consider any map $M: \{1,\dots,N\} \to X$ where $N \in \mathbb{N}$ and define the set $B(M)$ as: 
	
	$$ B(M) = \{S \in X^{\mathbb{N}} :  \ \text{for all} \ i \in \{1,\dots,N\} , \ S(i)= M(i) \}$$
	Let $\mathcal{B}_{X^{\mathbb{N}}}$ be the class of all such sets. Note that a for any $N \in \mathbb{N}$, the number of possible maps $M : \{1, \ldots N \} \to X$ is $|X|^N$. These sets are what can be interpreted as ``open balls" in $X^{\mathbb{N}}$.
	Let $\mathscr{T}_{X^{\mathbb{N}}}$ be the class of unions of arbitrary subcollections of $\mathcal{B}_{X^{\mathbb{N}}}$.   
	\begin{lemma}
		$\mathscr{T}_{X^{\mathbb{N}}}$ is the product topology on $X^{\mathbb{N}}$
	\end{lemma}
	\begin{proof}
		First, we show that $\mathscr{T}_{X^{\mathbb{N}}}$ is indeed a topology over $X^{\mathbb{N}}$. Notice that $\mathscr{T}_{X^{\mathbb{N}}}$ is closed under arbitrary unions by definition. To show that it is closed under finite intersections, let $ \bigcap_{i=1}^{K}B_i$ be a finite intersection such that $B_i \in \mathscr{T}_{X^{\mathbb{N}}}$ for all $i \in \{1, \ldots K\}$. Note that each $B_i$ is a union of some subcollection of $\mathcal{B}_{X^{\mathbb{N}}}$ and therefore we can write $B_i = \bigcup_{j_i \in \mathcal{J}_i}B_i^{j_i}$, with $\mathcal{J}_i$ being some indexed set, where each $B_i^{j_i}$ corresponds to an ``open ball" i.e. is a set of the form $B(M)$ for some $M : \{1, \ldots , N\} \to X$ and $N \in \mathbb{N}$.  Using the definition of $B(M)$, we know that there exist sets $A_1^i, A^i_2 \ldots $ with $A^i_j \subseteq X$ for all $j \in \mathbb{N}$ such that 
		$$B_i= \{S \in X^{\mathbb{N}}: \ \text{for all} \ j \in \mathbb{N}, S(j) \in A^i_j \}$$
		
		So, we can write $ \bigcap_{i=1}^{K}B_i$ as
		
		$$  \bigcap_{i=1}^{K}B_i = \{S \in X^{\mathbb{N}}: \ \text{for all} \ j \in \mathbb{N}, S(j) \in \bigcap_{i=1}^{K}A^i_j \}  $$

		
		Clearly, $  \bigcap_{i=1}^{K}B_i = \bigcup_{M \in \mathcal{M}_i}B(M)$ for some collection of maps $\mathcal{M}_i$. Therefore, $\mathscr{T}_{X^{\mathbb{N}}}$ is closed under finite intersection. Finally, $\mathscr{T}_{X^{\mathbb{N}}}$ contains $X^{\mathbb{N}}$ and $\varnothing$ as its elements. That $\varnothing\in\mathscr{T}_{X^{\mathbb{N}}}$ holds follows from the fact that $\varnothing$ is the union of elements from the empty subcollection of $\mathcal{B}_{X^{\mathbb{N}}}$. Further, $X^{\mathbb{N}}$ is the union of elements from the full collection $\mathcal{B}_{X^{\mathbb{N}}}$. Thus, $\mathscr{T}_{X^{\mathbb{N}}}$ is a topology over $X^{\mathbb{N}}$.\\
		\indent Now, we argue: $\Pi_{X^{\mathbb{N}}}\subseteq\mathscr{T}_{X^{\mathbb{N}}}$. For this, fix an arbitrary $i_\ast\in\mathbb{N}$ and $A\subseteq X$. If $A=\varnothing$, then $\pi_{i_\ast}^{-1}(A)=\varnothing$. As $\varnothing\in\mathscr{T}_{X^{\mathbb{N}}}$, $\pi_{i_\ast}^{-1}(A)\in\mathscr{T}_{X^{\mathbb{N}}}$ if $A=\varnothing$. However, if $A\neq\varnothing$, then observe:
		\begin{equation*}
			\pi_{i_\ast}^{-1}(A)=\bigcup\big\{B(M):M\in X^{\{1,\hdots,i_\ast\}}\;;\;M(i_\ast)\in A\big\}.
		\end{equation*}
		\noindent Thus, if $A\neq\varnothing$, then $\pi_{i\ast}^{-1}(A)\in\mathscr{T}_{X^{\mathbb{N}}}$. That is, $\pi_i^{-1}(A)\in\mathscr{T}_{X^{\mathbb{N}}}$ for every $A\subseteq X$. Hence, $\{\pi_i^{-1}(A):i\in\mathbb{N}\;;\;A\subseteq X\}\subseteq\mathscr{T}_{X^{\mathbb{N}}}$ and we have already shown that $\mathscr{T}_{X^{\mathbb{N}}}$ is a topology over $X^{\mathbb{N}}$. Further, by definition, $\Pi_{X^{\mathbb{N}}}$ is the smallest topology that satisfies $\{\pi_i^{-1}(A):i\in\mathbb{N}\;;\;A\subseteq X\}\in\Pi_{X^{\mathbb{N}}}$. Therefore, we obtain: $\Pi_{X^{\mathbb{N}}}\subseteq\mathscr{T}_{X^{\mathbb{N}}}$.\\
		\indent Finally, we argue: $\mathscr{T}_{X^{\mathbb{N}}}\subseteq\Pi_{X^{\mathbb{N}}}$. For this, fix an arbitrary $I\in\mathbb{N}$ and consider an arbitrary map $M:\{1,\hdots,I\}\rightarrow X$. For each $i\in\{1,\hdots,I\}$, let $A_i\coloneqq\{M(i)\}$. Then, we have the following:
		\begin{equation*}
			B(M)=\bigcap\big\{\pi_i^{-1}(A_i):i=1,\hdots,I\big\}.
		\end{equation*}
		\noindent Since $\Pi_{X^{\mathbb{N}}}$ is a topology and $\{\pi_i^{-1}(A):i\in\mathbb{N}\;;\;A\subseteq X\}\subseteq\Pi_{X^{\mathbb{N}}}$, it follows that $B(M)\in\Pi_{X^{\mathbb{N}}}$. Thus, $\Pi_{X^{\mathbb{N}}}$ is a topology over $X^{\mathbb{N}}$ such that $\mathcal{B}_{X^{\mathbb{N}}}\subseteq\Pi_{X^{\mathbb{N}}}$. Moreover, $\mathscr{T}_{X^{\mathbb{N}}}$ is the smallest topology over $X^{\mathbb{N}}$ such that $\mathcal{B}_{X^{\mathbb{N}}}\subseteq\mathscr{T}_{X^{\mathbb{N}}}$. Hence, we conclude: $\mathscr{T}_{X^{\mathbb{N}}}\subseteq\Pi_{X^{\mathbb{N}}}$.\\
	\end{proof}
	Now to show $d$ is a stopping rule if and only if it is continuous, first, assume that $d:X^{\mathbb{N}}\rightarrow X$ is continuous. Fix an arbitrary $S_\ast\in X^{\mathbb{N}}$ and let $y_{S_\ast}\coloneqq d(S_\ast)$. Now, we know that $\{y_{S_\ast}\}$ is open in the discrete topology over $X$. By continuity of the map $d$, the following set:
	\begin{equation*}
		d^{-1}(\{y_{S_\ast}\})\coloneqq\{S\in X^{\mathbb{N}}:d(S)=y_{S_\ast}\}
	\end{equation*}
	\noindent satisfies $d^{-1}(\{y_{S_\ast}\})\in\Pi_{X^{\mathbb{N}}}$. By the lemma above  and the definition of $\mathscr{T}_{X^{\mathbb{N}}}$, there exists $M:\{1,\hdots,k\}\rightarrow X$ such that $S_\ast\in B(M)\subseteq d^{-1}(\{y_{S_\ast}\})$. Now, $S_\ast\in B(M)$ implies: $M=S_\ast\vert_k$ and $B(M)=\{S_\ast\vert_k\cdot T:T\in X^{\mathbb{N}}\}$. Since $B(M)\subseteq d^{-1}(\{y_{S_\ast}\})$, it follows: $d(S_\ast\vert_k\cdot T)=y_{S_\ast}$ for all $T\in X^{\mathbb{N}}$. Since $y_{S_\ast}=d(S_\ast)$ and $S_\ast$ was arbitrary, we have established: if the map $d:X^{\mathbb{N}}\rightarrow X$ is continuous, then it is a stopping rule.\\
	\indent Now, assume that $d:X^{\mathbb{N}}\rightarrow X$ is a stopping rule. Since $X$ has the topology $2^X$, we must argue that $d^{-1}(A)\coloneqq\{S\in X^{\mathbb{N}}:d(S)\in A\}\in\Pi_{X^{\mathbb{N}}}$ for any $A\subseteq X$. Since $d^{-1}$ preserves arbitrary unions and $\Pi_{X^{\mathbb{N}}}$ is closed under arbitrary unions, it is enough to argue that $d^{-1}(\{y\})\in\Pi_{X^{\mathbb{N}}}$ for any $y\in X$. So, fix an arbitrary $y_\ast\in X$. If $d^{-1}(\{y_\ast\})=\varnothing$, then we have nothing more to argue as $\varnothing\in\Pi_{X^{\mathbb{N}}}$. Hence, assume that $d^{-1}(\{y_\ast\})\neq\varnothing$. Consider an arbitrary $S_\ast\in d^{-1}(\{y_\ast\})$. Since $d$ is a stopping rule, there exists $k(S_\ast)\in\mathbb{N}$ such that: $d(S)=y_\ast$ for every $S\in B(S_\ast\vert_{k(S_\ast)})$. This is because $B(S_\ast\vert_{k(S_\ast)})=\{S_\ast\vert_{k(S_\ast)}\cdot T:T\in X^{\mathbb{N}}\}$. Thus, we have:
	\begin{equation*}
		\bigcup\big\{B(S\vert_{k(S)}:S\in d^{-1}(\{y_\ast\}))\big\}=d^{-1}(\{y_\ast\}).
	\end{equation*}
	\noindent Hence, $d^{-1}(\{y_\ast\})\in\mathscr{T}_{X^{\mathbb{N}}}$ by definition of $\mathscr{T}_{X^{\mathbb{N}}}$. By the lemma above, it follows that $d^{-1}(\{y_\ast\})\in\Pi_{X^{\mathbb{N}}}$. Since $y_\ast\in Y$ was arbitrary, we have: $d^{-1}(A)\in\Pi_{X^{\mathbb{N}}}$ for any $A\in 2^X$. Thus, if the $d:X^{\mathbb{N}}\rightarrow X$ is a stopping rule, then it is continuous. \qed
	
	\begin{center}
		\textsc{Proof of Theorem 4}
	\end{center}
	
	We first prove the necessity. Suppose $d$ is a Threshold based Stopping rule  with $v \in \mathbb{R}_+$ and $w : X \to \mathbb{R}_{++}$. Consider an arbitrary $S \in X^{\mathbb{N}}$ with $d(S)=x$. Then there exists $N_1 \in \mathbb{N}$ such that $W^{N_1}_S(x) \geq v >  W^{N_1}_S(y)$ for all $y \neq x$. Consider any $S' \in \mathcal{FS}(S,x)$ where $S'$ and $S$ differ on $k$ and $k+1^{\text{th}}$ position for some $k\in \mathbb{N}$. Suppose $N_1 = k+1 $ and $S(k) \neq x$ (and $S(k+1)=x$). In this case, we get  $W^{N_1-1}_{S'}(x) \geq v >  W^{N_1-1}_{S'}(y)$ for all $y \neq x$ and we have $d(S') =x$. For all other cases, we have $W^{N_1}_{S'}(x) \geq v >  W^{N_1}_{S'}(y)$ and therefore we get $d(S') =x$. A similar argument holds for any $S' \in \mathcal{FD}(S,x)$. Therefore $d(S) = d(S')$ for all $S' \in \mathcal{F}(S,x)$. Since $S$ was arbitrary, we have shown that $d$ satisfies Monotonicity. To show that $d$ satisfies Informational Dominance, consider a minimal sufficient segment $M$ with $d(M)=x$ and a sufficient segment $N$ such that $d(N)=y \neq x$ and $x \notin N$. Assume for contradiction that $d([M' \cdot N]) =x$ for some segment $M'$ such that $M \triangleright M'$. Since $M'$ is not sufficient, the cumulative weight of $x$ in $M'$ is less than $v$. Since $x \notin N$, the cumulative weight of $x$ in $[M'\cdot N]$ is less than $v$, a contradiction. Therefore $x \neq d([M'\cdot N])$.			
	
	Now, we prove the sufficiency. Let $d$ satisfy Monotonicity and Informational Dominance. First, we construct the ``revealed" critical frequency of each alternative. Fix $x \in X$. Note, by definition of a choice rule, $d(S^x)=x$ for the constant sequence $S^x = (x, x \ldots)$. Since $d$ is a stopping rule, there exists $k \in \mathbb{N}$ such that $d(S^x)= d([S^x|_k] \cdot T)$ for all $T \in X^{\mathbb{N}}$. Let $n_x \coloneqq \inf \{k \in \mathbb{N} : d(S^x)= d([S^x|_k] \cdot T) \ \forall T \in X^{\mathbb{N}}\}$. Since $\mathbb{N}$ is well-ordered, we know that $n_x \in \mathbb{N}$. Consider an arbitrary sequence $S$ with $d(S)=x$. For any $i \in \mathbb{N}$ and the segment $S|_i$ of $S$, denote by $\# x(S|_i)$ the number of appearances of $x$ in it. That is, $ \# x(S|_i) \coloneqq |\{j \in [i]: [S|_i](j) =x\}|$. Now, denote by $i(S,a)$ the position at which an alternative $a$ reaches $n_a$ appearances in $S$. That is, $i(S,a) \coloneqq \inf\{i\in \mathbb{N}: \#a(S|_i) = n_a \} $. In case alternative $a$ does not reach $n_a$ appearances in $S$, let $ i(S,a) = \infty$ with the convention that $n < \infty$ for all $n \in \mathbb{N}$. 
	
	We show that $i(S,x) < i(S,y)$ for all $y \neq x$. Assume for contradiction that $i(S,y)< i(S,x)$ for some $y \neq x$. Let $S'$ be a sequence generated from $S$ by deleting all the first terms in the first $i(S,x)$ positions that are not equal to $x$ or $y$. That is, $S'$ is generated by finitely many favorable deletions with respect to $x$ and $y$. By Monotonicity, we have $d(S')=x$. Note that first $i(S',x)$ terms contain $n$ number of $y$'s and $n_x$ number of $x$'s ($n + n_x = i(S',x)$ ) where $n \geq n_y$. Now, consider finitely many favorable shifts of $S'$ with respect to $x$ to generate $S''$ such that its first $n_x$ terms are all $x$ followed by $n$ terms that are $y$. Again, by Monotonicity, we have $d(S'')=x$. 
	Denote this initial segment of $x$'s as $[M \cdot x]$ where $M$ is  $(n_x-1)$ long segment of  $x$'s  and the $n$ long segment of $y$'s as $N$. So, we can write $S'' = [Mx\cdot N]\cdot T $ where $T \in X^{\mathbb{N}}$ and $T(j) = S'(i(S',x) +j)$ for all $j \in \mathbb{N}$.	By the definition of $n_x$ we know that there exists some $T \in X^{\mathbb{N}}$ such that $d(M\cdot T) \neq x$. Also, by the definition of $n_y$, we know that $d(N \cdot T) =y$ for all $T \in X^{\mathbb{N}}$. In other words, $Mx$ is a minimal sufficient segment, $M$ is not a minimal sufficient segment and $N$ is a sufficient segment. Using Informational Dominance, we know that $d([M \cdot N ]) \neq x$. It must be that $d([M \cdot N \cdot x \cdot T])=y$ for all $T \in X^{\mathbb{N}} $. Suppose not i.e. $d(M \cdot N \cdot x \cdot T)=z$ for some $z \neq x,y$ and $T \in X^{\mathbb{N}}$. Then, by Monotonicity, it must be that $d(NxT)= z$, a contradiction since $N$ contains $n_y$ first $y$'s. Therefore $d(M \cdot N \cdot x \cdot T)=y$ for all $ T \in X^{\mathbb{N}}$. Now, notice that we can generate the sequence $S'$ by successively moving $y$'s to the left i.e. finitely many favorable shifts with respect to $y$ and again, by Monotonicity, we have $d(S') = y$, a contradiction. Therefore, we get $i(S,x) < i(S,y)$. 
	
	Now, we define $w(x) \coloneqq \frac{1}{n_x}$ for all $x \in X$ and let $v =1$. Consider the computable choice rule $d^*$ such that $d^*(S)= \{x : W^N_S(x) \geq v > W^N_S(y)\} $
	for all $S \in X^{\mathbb{N}}$. Note that since $w(x) > 0$, $d^*$ is indeed a computable choice rule. We will show that $d^* = d$. Consider any arbitrary $S \in X^{\mathbb{N}}$ and let $d(S)= z$. We know that $i(S,z) < i(S,y)$ for all $y \neq z$. Let $i(S,z) = N$. By construction, we know that $W_S^N(z) =1$ and  $ W_S^N(z) \geq v > W_S^N(y)$ for all $y \neq z$ and therefore $d^*(S)=z$. Since $S$ was chosen arbitrarily, we have shown that $d^* = d$.	 \qed

	\begin{center}
		\textsc{Proof of Theorem 5}
	\end{center} 
	
	We will first show the necessity part. Suppose $d$ is OSR with $(\succ, a^*,k)$. First, note that $D = U(a^*)$ and length of a minimal sufficient segment $M$ is $k$ if $M \subset D'$. Also, the length of any minimal sufficient segments is not greater than $k$. Since $ \succ $ is a preference order, $d$ satisfies Sequential-$\alpha$ and Sequential-NBC. To show that $d$ satisfies Replacement, consider any arbitrary $M \in \mathscr{MS}_{D'}$  and $M'  \in \mathcal{S}_m$ such that $M(i)= x$, $M'(i)=y$ for some $i \in [m]$ and $M(j)=M'(j)$ for all $j \neq i$. W.L.O.G let $z$ be the $\succ$-maximal alternative in $M'$. For any $S \in X^{\mathbb{N}}$, we know that $d(M' \cdot S ) = z$ and therefore $M \in \mathscr{S}$ 
	
	Now, to show the sufficiency, suppose $d$ satisfies Sequential-$\alpha$, Sequential-NBC and Replacement. We will proceed in several steps as follows: 
	
	\textit{Step 1}: First, we show that $|M|= |M'|$\footnote{We abuse notation and denote the length of the segment $M$ by $|M|$.} for any $M, M' \in \mathscr{MS}_{D'}$. Suppose not. W.L.O.G let $|M| > |M'|$. Consider the segment $M|_{|M'|}$. Since $X$ is finite, we can reach from $M'$ to $M|_{|M'|}$ in finite number of ``steps" of replacement i.e. there exists a chain of segments $M_1, \ldots , M_n$ with $M_1 = M'$ and $M_n = M|_{|M'|}$ such that $|\{i : M_j(i) \neq M_{j+1}(i)\}| = 1$ for all $j \in \{1 , \ldots n-1\}$. In other words, $M_j$ and $M_{j+1}$ differ only in one position for all segments in the chain. By repeated application of Replacement, we know all the segments in the chain are sufficient and therefore $M|_{|M'|} \in \mathscr{S}$. Since $|M_{|M'|}| < |M|$ and $[M|_{|M'|}](i) = M(i)$ for all $i$, this is a contradiction to $M \in \mathscr{MS}$. We have established that all minimal sufficient segments that do not contain any decisive alternatives are of the same length. Let that length be denoted by $i^{D'}$. 
	
	\textit{Step 2}: Consider an arbitrary $M \in \mathscr{MS_{D}}$ and let $i^D \coloneqq \inf \{i \in \mathbb{N} : M(i) \in D\}$. That is, $i^D$ denotes the location of first occurrence of a decisive alternative in $M$. We will show that $|M| \leq i^{D'}$ and $|M| = i^D$. Suppose not i.e. there exists a $M \in \mathscr{MS}_{D}$ such that $|M| > i^{D'}$. Consider an arbitrary $M' \in \mathscr{MS}_{D'}$. By the previous step, we know that $|M'| = i^{D'}$. Since $X$ is finite, as in the previous step, consider a chain of segments $M_1 \ldots ,  M_n$ such that $M_1 = M'$ and $M_n = M|_{|M'|}$ such that every successive element in the chain differs by an alternative in exactly one position. By repeated application of Replacement, we know that $M_n$ is a sufficient segment. Since $|M_n| < M$ and $M_n(i) = M(i)$ for all $i$, this is a contradiction to $M \in \mathscr{MS}$. Therefore, $|M| \leq i^{D'}$. Now, we will show that $|M| = i^D$. Assume for contradiction that $|M| > i^D$ (note that the argument for the case $|M| < i^D$ is trivial by the definition of $\mathscr{MS}_{D}$). W.L.O.G, let $M(i^D)=x$. By the definition of $D$, we know that $d(M)=x$. Since $M \in \mathscr{MS}$, there exists a sequence $T$ such that $d([M|_{|M|-1}\cdot T]) \neq d(M)=x$. Let $\bar{M}$ be the minimal sufficient segment of the  sequence $[M|_{|M|-1}\cdot T]$. By the definition of $D$, $|\bar{M}| < i^D$. Since $|\bar{M}| <|M|$ and $\bar{M}(i) = M(i)$ for all $i$, we have a contradiction to $M \in \mathscr{MS}$. Therefore, we must have $|M| = i^D$.   
	
	\textit{Step 3}: Since $d$ is a computable rule (and therefore a stopping rule), it is completely specified by the choices on its minimal sufficient segments. By the previous two steps, we know that the length of any minimal sufficient segment is at most $i^{D'}$. Let $k = i^{D'}$. There are two possible cases:
	
	\noindent (i)  $k=1$. This implies that size of all minimal sufficient segment is 1. Further, this implies that every alternative is decisive i.e. $D =X$. Consider an arbitrary preference order $\succ$ on $X$ and let $a^* \coloneqq \min(X, \succ)$.\footnote{$\min (X, \succ) \coloneqq \{x \in X : y \succ x \ \forall y \in X\}$} It is easy to see that $(k,\succ, a^*)$ rationalize $d$. 
	
	\noindent (ii) $k \geq 2$ . Consider any $x,y \in D'$. Define $ \succ$ as follows: $x \succ y $ iff there exists $M \in \mathscr{MS}_{D'}$ such that $ x,y \in M$ and $d(M)=x$. We first show that $ \succ$ is a preference order over $D'$. Reflexivity follows from the definition. Assume for contradiction that $\succ$ is not antisymmetric. Then there exists distinct $x,y \in D'$ such that $x  \succ y$ and $y \succ x$. By definition, there exists $M , M' \in \mathscr{MS}_{D'}$ such that $x,y \in M$, $x,y \in M$, $d(M)=x$ and $d(M')=y$
	Consider $M'' \in \mathscr{MS_{D'}}$ such that $x,y \in M''$ and $z \notin M''$ for all $z \neq x,y$. That is, $M''$ consists of only $x$ and $y$ (such $M''$ exists due to the assumption that $k \geq 2$ and step 1). By Sequential-$\alpha$,  we have $d(M'')=x$ and $d(M'')=y$, a contradiction. Therefore $\succ$ is antisymmetric.  Now, consider any distinct $x,y \in D'$ and $M \in \mathscr{MS}_{D'}$ such that $x,y \in M$ and $z \notin M$ for all $z \neq x,y$. By definition of $\succ$ and antisymmetry we have either $x \succ y$ or $y \succ x$. Therefore $\succ$ is complete. To show $\succ$ is transitive, consider $x,y,z  \in D'$ and suppose $x \succ y$ and $y \succ z$. We consider two cases: 
	
	\begin{itemize}
		\item[(a)]$k=2$. We know that there exists $M,M' \in \mathscr{MS}_{D'}$ such that $y \in M$, $d(M)=x$ and $z \in M'$, $d(M')=y$. Consider $M'' \in \mathscr{MS}_{D'}$ such that $x,z \in M''$. Since $|M''|=2$, by Sequential-NBC, we know that $d(M'') \neq z$, implying $d(M'')=x$ giving us $x \succ z$.  
		\item[(b)] $k>2$. Consider $M \in \mathscr{MS}_{D'}$ such that $x,y,z \in M$ and $w \notin M$ for all $w \neq x,y,z$. By Sequential-$\alpha$, we know that $d(M) \neq z$ and $d(M) \neq y$. Therefore $d(M)= x$ implying $x \succ z$. 
	\end{itemize}   
	
	\textit{Step 4}: We have shown that $\succ$ is a preference order over $D'$. Now, we show that $D$ is non-empty. Assume for contradiction that $D$ is empty i.e. $X = D'$. By step 1, all minimal sufficient segments are of the same length. Since $\succ$ is a preference order over $X$, we have a unique maximal element. W.L.O.G let it be $x$. Consider an arbitrary minimal sufficient segment $M$ such that $x \in M$. Since $\succ$ is antisymmetric, we know that $d(M)=x$. Therefore, by definition of $D$, we must have $x \in D$, a contradiction.  Now, consider an arbitrary preference order  $ \bar{\succ}$ on $X$ such that $\succ \subset \bar{\succ}$ and $x \ \bar{\succ} \ y$ for all $x \in D$ and $ y \in D'$. 
	Let $a^* = \min(D,\bar{\succ})$. Now, we will show that $(k, \bar{\succ}, a^*)$ rationalize $d$. Consider an arbitrary $S \in X^{\mathbb{N}}$ with $d(S)=x$. There are two possible cases: (i) The segment of $S_k$ does not contain any alternative from $D$. That is $S|_k \in \mathscr{MS}_{D'}$. Suppose there exists $y \in S|_k$ with $y \neq x$ such that $y \succ x$. Then we have a contradiction to the antisymmetry of $\succ$. Therefore, by completeness of $\succ$ we have $x \succ y$ for all $ y \in S|_k$ (ii) The segment $S|_k$ contains at least one alternative from $D$. That is $M \in \mathscr{MS}_{D}$ for some $S|_k \triangleright M$. By step 2, we must have $x \in D$ and $x$ is the first alternative from $D$ to feature in $S|_k$ and we are done.  \qed

	\setlength{\bibsep}{0.2 cm} 
	\bibliographystyle{ecta} 
	\bibliography{bibliog.bib} 	
	
\end{document}